\documentclass[11pt]{article}
\usepackage[utf8]{inputenc}

\usepackage{authblk}
\usepackage{amsthm, amssymb, amsmath, mathtools, mathrsfs}
\usepackage{bm,xcolor}
\usepackage{physics}
\usepackage{geometry, graphicx, subcaption, epstopdf, float, fullpage}
\usepackage{qcircuit}
\usepackage{cite}
\usepackage{enumitem}
\usepackage{tikz-cd}

\usepackage[bookmarksnumbered,colorlinks,citecolor=red,linkcolor=blue,hyperindex,linktocpage=true]{hyperref}
\usepackage[nameinlink,capitalise]{cleveref}
\usepackage{appendix}

\newtheorem{theorem}{Theorem}[section]
\newtheorem{corollary}{Corollary}[theorem]
\newtheorem{lemma}{Lemma}
{\bf}{\it}
\newtheorem{remark}[theorem]{Remark}
\newtheorem{definition}[theorem]{Definition}
\newtheorem{Prop}[theorem]{Proposition}

\newcommand{\Anc}{\mathcal{H}_{\mathcal{A}}}
\newcommand{\IAd}{I_{\mathcal{A}} }

\newcommand{\Ran}{\operatorname{Range}}
\newcommand{\tH}{\widehat H}
\newcommand{\pH}{\widehat H_{\operatorname{proj}}}

\def\R{\mathbb{R}}
\def\C{\mathbb{C}}
\def\H{\mathbb{H}}
\def\V{\mathbb{V}}
\def\X{\mathbb{X}}
\def\M{\mathbb{M}}

\def\bu{\bm u}

\newcommand{\mat}[1]{\begin{pmatrix}#1 \\ \end{pmatrix}}

\title{Quantum Simulation of Differential-Algebraic Equations with Applications to Unsteady Stokes Flow}
\author[1]{Hsuan-Cheng Wu\thanks{wu.hsuancheng@psu.edu}}
\author[1]{Xiantao Li\thanks{Xiantao.Li@psu.edu}}

\affil[1]{Department of Mathematics, The Pennsylvania State University, PA 16802, USA}

\date{\today}
\begin{document}

\maketitle

\begin{abstract}
Differential-algebraic equations (DAEs) arise naturally in constrained dynamical systems,
but their algebraic constraints and hidden compatibility conditions make them more subtle than
standard ordinary differential equations. This paper initiates a quantum-algorithmic study of
constrained linear DAEs. We introduce a dilation framework that embeds the generally non-Hermitian
constrained evolution into a projected Schr\"odinger-type dynamics on an enlarged Hilbert space,
\[
i\frac{d}{dt}\Psi(t)=P\widehat H P\Psi(t),
\]
where $\widehat H$ is Hermitian and $P$ is the orthogonal projector onto the lifted constraint
subspace. This identifies the DAE evolution with a quantum Zeno-type dynamics and enables the use
of block encodings, QSVT-based projector construction, and Hamiltonian simulation.

We apply the framework to structure-preserving discretizations of the unsteady Stokes equations,
where the pressure enforces the discrete incompressibility constraint.  For Stokes, the Zeno-reduced generator has the projected square factorization
\[
    S_h=-\Pi_h\Delta_h\Pi_h=(G_h\Pi_h)^\dagger(G_h\Pi_h),
\]
which can be represented through a Gaussian moment dilation and implemented as a Gaussian
superposition of unitary Zeno evolutions generated by a first-order square-root Hamiltonian.  In the generic sparse-access
model, this gives a simulation-stage cost $\widetilde O(h^{-2}\sqrt t)$, up to the usual
postselection factor for preparing the normalized dissipative state.  The results provide a
first step toward understanding the intersection of quantum algorithms, DAEs, constrained PDE
dynamics, and square-root Gaussian dilations.
\end{abstract}

\section{Introduction}
Differential-algebraic equations (DAEs) arise naturally in the mathematical modeling of constrained dynamical systems and appear in many applications, including constrained mechanics, fluid dynamics, and multibody systems, \cite{10.1115/1.3167743, paraskevopoulos2017augmented, muzychka2010unsteady, huray2009maxwell}. A linear DAE is a system of equations of the form
\begin{equation} \label{eq: general DAE}
    E \bm x'(t) = A \bm x(t) + \bm f(t),
\end{equation}
where the matrix $E$ may be singular. In contrast to ordinary differential equations, not all components of $\bm x(t)$ are governed by differential laws; some are constrained by algebraic equations that must hold for all time. As a result, DAEs describe dynamics evolving on a constraint manifold, and their analysis is generally more subtle than that of standard ODEs. In particular, issues such as consistency of initial data, hidden constraints, and index play a central role in both theory and computation  \cite{wanner1996solving, ascher1991projected, petzold1982differential, RABIER2002183, gear1971simultaneous, gear1990differential}. A central example is the dynamics of incompressible flow, where the velocity is governed by an evolution equation while the pressure acts as a Lagrange multiplier enforcing the divergence-free constraint. A specific example is the unsteady Stokes system,
\begin{equation}
    \bu_t - \Delta \bu + \nabla p = \bm f,
    \qquad
    \nabla \cdot \bu = 0,
\end{equation}
where the incompressibility condition is algebraic in nature. After spatial discretization, one obtains a linear DAE in which the pressure variable does not evolve through its own differential equation, but instead enforces a constraint on the admissible velocity field. This structure makes DAEs fundamentally different from standard evolution equations and creates significant challenges for both analysis and numerical simulation.

Recent developments in quantum algorithms open up a new direction for scientific computing of large-scale problems. 
A central primitive in quantum computing is Hamiltonian simulation, namely the implementation of the unitary evolution
\(
e^{-itH}
\)
generated by a Hermitian matrix $H$ \cite{childs2012hamiltonian, berry2014exponential, gilyen2019quantum, childs2021theory, an2021time, ALL23, dizaji2024hamiltonian, eldar2002quantum}. Hamiltonian simulation lies at the core of many quantum algorithms, including those for linear systems \cite{harrow2009quantum, huang2019near, pan2014experimental, dervovic2018quantum, clader2013preconditioned}, differential equations \cite{childs2020quantum, liu2021efficient, joseph2020koopman, Ber14}, and quantum machine learning \cite{biamonte2017quantum, schuld2015introduction, cerezo2022challenges, huang2021power, ivashkov2026qkan}. 

Despite recent developments in quantum algorithms for differential equations and linear systems,
quantum algorithms for DAEs with explicit complexity analysis remain relatively unexplored. 
First,  a linear DAE is not, in general, presented in a form directly amenable to Hamiltonian simulation. Second, the presence of algebraic constraints means that the dynamics are confined to a proper subspace, and evolving the system without respecting that constraint would destroy the physical solution. Therefore, to use quantum hardware effectively, one needs a formulation that both preserves the constraint and converts the dynamics into a unitary form. 

The main idea of this work is to achieve such a reformulation through a dilation procedure \cite{li2025linear, jin2023schrodingerisation_pra, jin2022quantum}. Starting from a constrained linear DAE, we embed the dynamics into a larger Hilbert space by introducing an ancillary degree of freedom. This dilation transforms the original non-Hermitian constrained evolution into a Schr\"{o}dinger-type equation on the enlarged space. More importantly, the constraint can be lifted to a subspace condition in the dilated space, so that the physical evolution is recovered by restricting the dynamics to the corresponding invariant subspace. The resulting equation takes the form of a projected Hamiltonian evolution,
\begin{equation}
    \Psi'(t) = -i P \tH P \Psi(t),
\end{equation}
where \(P\) is the orthogonal projector onto the constrained subspace and \(\tH\) is a Hermitian operator on the dilated space. Remarkably, this is precisely the structure associated with quantum Zeno dynamics \cite{dizaji2024hamiltonian, facchi2008quantum, hahn2022unification, itano1990quantum, pascazio1994dynamical, koshino2005quantum}. In other words, the constrained DAE can be represented as the effective dynamics generated by repeated projection of an unconstrained unitary evolution. This observation provides a direct conceptual bridge between DAEs and Hamiltonian simulation.

The quantum Zeno form is particularly valuable because it leads to concrete quantum simulation strategies. One may either simulate the projected Hamiltonian \(P \tH P\) directly, provided a suitable block-encoding is available, or realize the same evolution through repeated short-time evolutions under \(\tH\) interlaced with projective restriction to the constraint subspace. In both viewpoints, the essential difficulty is shifted from the original DAE structure to the construction of the projection and the lifted Hamiltonian. Once this is achieved, the problem falls within the standard framework of unitary quantum dynamics. Hence the dilation does not merely provide an abstract embedding; it converts a constrained non-unitary evolution into an object compatible with the basic architecture of quantum algorithms.

A second theme of this paper is that standard spatial discretizations of the Stokes equations
naturally fit into this framework. In particular, mixed finite element discretizations lead to
semidiscrete saddle-point systems of the form
\begin{equation}
    \mat{M_h & 0 \\ 0 & 0}
    \mat{\bm u_h' \\ \bm p_h'}
    =
    \mat{-L_h & -D_h^\dag \\ D_h & 0}
    \mat{\bm u_h \\ \bm p_h}
    +
    \mat{\bm f_h \\ 0},
\end{equation}
where $M_h$ is the velocity mass matrix, $L_h$ is the velocity stiffness matrix, and $D_h$ is the
discrete divergence operator. This is a constrained linear DAE: the differential equation governs
the velocity unknowns, while the pressure appears as a Lagrange multiplier enforcing the discrete
incompressibility constraint. The saddle-point structure therefore provides a natural algebraic
setting for the dilation theory. The same philosophy also applies to staggered-grid finite
difference methods such as the Marker-and-Cell (MAC) scheme, where the discrete divergence and
gradient are arranged in a compatible way. Both MAC discretizations and mixed finite element
methods produce algebraic systems that separate the differential part from the constraint part,
making them well suited for the projected Hamiltonian formulation.

For the Stokes problem, the Zeno reduction reveals an additional structure beyond the generic DAE
formulation.  On the divergence-free subspace, the reduced generator is a projected positive square,
\[
    S_h=-\Pi_h\Delta_h\Pi_h=(G_h\Pi_h)^\dagger(G_h\Pi_h).
\]
This permits a more efficient square-root formulation.  Instead of applying the general
moment-matching dilation directly to the diffusive generator, one introduces the projected
first-order Dirac Hamiltonian
\[
    B_h=
    \begin{pmatrix}
        0&\Pi_hG_h^\dagger\\
        G_h\Pi_h&0
    \end{pmatrix},
    \qquad
    B_h^2=
    \begin{pmatrix}
        S_h&0\\
        0&*
    \end{pmatrix}.
\]
The heat semigroup $e^{-tS_h}$ is then obtained by a Gaussian moment dilation of
$e^{-x^2}$, equivalently by a Gaussian superposition of unitary evolutions generated by $B_h$.
This connects the Zeno formulation of the constraint with Gaussian-LCHS and transmutation
methods for factorizable dissipative dynamics~\cite{kharazi2026sublinear,jin2026transmutation}.

The purpose of this paper is therefore threefold. First, we develop a dilation framework for
constrained autonomous linear DAEs and show that it yields an exact projected Schr\"odinger
dynamics of quantum Zeno type. Second, we apply this framework to incompressible Stokes systems
arising from structure-preserving spatial discretizations, with particular attention to
staggered-grid formulations and mixed finite element methods. Third, for the Stokes system, we
show that the Zeno-reduced square structure can be combined with a Gaussian moment dilation to
improve the simulation-stage dependence on the diffusive scale. We also discuss the resulting
complexity estimates and give a heuristic comparison with standard classical projection methods,
emphasizing the dependence on the mesh size. This comparison is not intended as a definitive
end-to-end separation, but rather as a transparent starting point for identifying where quantum
simulation techniques may offer asymptotic advantages and where additional overheads, such as
state preparation and measurement, must be accounted for.

While this manuscript was being completed, Dutt et al. posted an independent quantum DAE solver for simulating RLC circuit dynamics based on modified nodal analysis~\cite{DuttChowdhuryTemmeKrovi2026}; our work differs in its quantum-Zeno dilation formulation and in its application to constrained PDE dynamics, particularly semidiscrete Stokes systems.


The remainder of the paper is organized as follows. In \cref{section: math_form}, we present the
abstract constrained linear DAE and derive its dilation into a projected Schr\"odinger form. We
then explain how this leads to a quantum Zeno interpretation and discuss the corresponding
simulation mechanism. In \cref{sec: Stokes}, we apply the framework to the unsteady Stokes
equations, first through Marker-and-Cell (MAC) staggered-grid finite differences and then through
mixed finite element discretizations. We then compare the direct projected-Hamiltonian simulation
baseline with a Gaussian-Zeno square-root construction that exploits the factorization of the
reduced Stokes operator.

\section{Preliminaries} \label{sec: prelim}
The implementation of quantum algorithms requires an appropriate input model. 
We briefly review the standard quantum algorithm primitives that will be used later to implement the projected Hamiltonian and the constraint projector. We use block-encoding techniques and quantum singular value transformation (QSVT) \cite{gilyen2019quantum} as the primary building blocks.

\begin{definition}[Block-encoding]
Let $A\in\mathbb{C}^{n\times n}$, let $\alpha\ge \|A\|$, and let $a\in\mathbb{N}$. A unitary
$U\in\mathbb{C}^{2^a n\times 2^a n}$ is called an $(\alpha,a,\varepsilon)$ block-encoding of $A$ if
\begin{equation}
\left\|
A-\alpha(\langle 0^a|\otimes I_n)\,U\,(|0^a\rangle\otimes I_n)
\right\|\le \varepsilon.
\end{equation}
\end{definition}

In other words, after preparing the ancilla register in the state $|0^a\rangle$, the upper-left block
of $U$ is $A/\alpha$ up to error $\varepsilon$. This representation allows one to encode a nonunitary matrix into a unitary.

\begin{lemma}[Some basic properties of block--encodings]
\leavevmode
\begin{enumerate}
    \item[(a)] Any unitary $U$ is a $(1,0,0)$ block-encoding of itself.
    \item[(b)] If $U$ is an $(\alpha,a,\varepsilon)$ block-encoding of $A$, then $U^\dagger$ is an
    $(\alpha,a,\varepsilon)$ block-encoding of $A^\dagger$.
\end{enumerate}
\end{lemma}
\begin{lemma}[Product and Linear Combinations of block--encodings \cite{childs2012hamiltonian}]
\label{lem:BE-lcu}
Let $A$ and $B$ have an $(\alpha,n,\varepsilon)$ block encoding with gate complexity $T_A$ and a $(\beta,m,\delta)$ block encoding with gate complexity $T_B$, correspondingly. Then 
\begin{enumerate}
    \item[(a)] There exists an $(\alpha+\beta,n+m,\alpha\delta+\beta\varepsilon)$ block encoding of $A+B$ with gate complexity $O(T_A+T_B)$.
    \item[(b)] There exists an $(\alpha\beta,n+m,\alpha\delta+\beta\varepsilon)$ block encoding of $AB$ with gate complexity $O(T_A+T_B)$.
\end{enumerate}
\end{lemma}

The algebraic constrains in a DAE corresponds to the null space of a matrix $A$. A particularly useful quantum algorithm for approximating the projection to the null spce is the Quantum singular value transformation (QSVT) 
 \cite{gilyen2019quantum}.
\begin{definition}[Quantum singular value transformation \cite{gilyen2019quantum}]
Let $A\in\mathbb{C}^{m\times n}$ and let $U$ be a block-encoding of $A/\alpha$.
Assume that
\begin{equation} A=\sum_j \sigma_j |u_j\rangle\langle v_j| \end{equation}
is a singular value decomposition, with singular values $\sigma_j\in[0,\alpha]$. QSVT is a procedure that, given a polynomial $p$ satisfying suitable parity and boundedness conditions, constructs a unitary whose action on the
singular subspaces of $A$ applies the map
\begin{equation}
\sigma_j/\alpha \;\longmapsto\; p(\sigma_j/\alpha).
\end{equation}
Equivalently, QSVT implements a block-encoding of a matrix function of $A$ whose singular values
are transformed by the polynomial $p$.
\end{definition}

For Hermitian matrices, QSVT reduces to polynomial transformation of singular values. This is the
mechanism behind Hamiltonian simulation, spectral filtering, and projector construction.

\begin{lemma}[Polynomial transformation by QSVT, {\protect\cite[Corollary~11]{gilyen2019quantum}}]
Let $U$ be an $(\alpha,a,\varepsilon)$ block-encoding of a matrix $A$, and let
$p:[-1,1]\to\mathbb{C}$ be a polynomial of degree $d$ satisfying the standard QSVT admissibility
conditions (in particular, appropriate parity and $|p(x)|\le 1$ on $[-1,1]$).
Then there exists a quantum circuit using $q$ queries to $U$ and $U^\dagger$ that implements a
block-encoding of the matrix obtained by replacing each singular value $\sigma_j/\alpha$ of
$A/\alpha$ by $p(\sigma_j/\alpha)$, up to error determined by the block-encoding error and the
polynomial approximation error.
\end{lemma}

A particularly important application for us is the construction of approximate spectral projectors.

\begin{lemma}[Approximate projector from a spectral gap, {\protect\cite[Theorem~19]{gilyen2019quantum}}]
Let $A$ be Hermitian with spectrum contained in $\{0\}\cup [\gamma,\alpha]$ for some
$0<\gamma\le \alpha$, and suppose that $U$ is a block-encoding of $A$.
Then for every $\varepsilon\in(0,1)$ there exists a polynomial $p$ of degree
\begin{equation}
q=O\!\left(\frac{\alpha}{\gamma}\log\frac{1}{\varepsilon}\right)
\end{equation}
such that
\begin{equation}
p(0)\approx 1, \qquad |p(x)|\le \varepsilon \quad \text{for } x\in[\gamma/\alpha,1],
\end{equation}
and hence QSVT yields an $\varepsilon$-approximate block-encoding of the orthogonal projector
onto $\ker(A)$.
\end{lemma}

We then introduce the complexity of optimal Hamiltonian simulation, which is proved in \cite{low2019hamiltonian, gilyen2019quantum}.
\begin{lemma}[Hamiltonian simulation from block-encoding, {\protect\cite[Corollary~32]{gilyen2019quantum}}] \label{lem: optimal sim}
Let $H$ be Hermitian and suppose $U_H$ is an $(\alpha,a,\varepsilon')$ block-encoding of $H$.
Then for any $t\ge 0$ and target precision $\varepsilon\in(0,1)$ ($\varepsilon'=\varepsilon/t$), one can implement a unitary
approximating $e^{-itH}$ with query complexity
\begin{equation}
O\!\left(
\alpha t + \frac{\log(t/\varepsilon)}
{\log\!\bigl(e+\frac{\log(t/\varepsilon)}{\alpha t}\bigr)}
\right),
\end{equation}
up to standard overheads in ancilla qubits and elementary gates.
\end{lemma}

\paragraph{Notation}
In this work, we denote the vectors by bold lower case, e.g., $\bm x$ and $\bm u$. We denote quantum states in ket-bra notations. $\dag$ represents the conjugate transpose, and $'$ is the derivative with respect to time.

\section{Mathematical Formulation} \label{section: math_form}

To facilitate the quantum simulation of DAE systems, we develop the formal framework for mapping a constrained autonomous linear differential-algebraic equation to a dilated Schr\"{o}dinger-like system. We then demonstrate how this dilation naturally leads to a quantum Zeno-like projected evolution. 

\subsection{From DAE to Dilated DAE}

Consider a constrained autonomous linear DAE for the state vector $\bm x(t) \in \mathbb{C}^n$:
\begin{subequations} \label{eq:original_dae}
\begin{align}
    & \frac{d}{dt} \bm x(t) = L \bm x(t) + C^{\dagger} \bm \lambda(t), \label{eq:dae_ode} \\
    & C \bm x(t) = 0, \label{eq:dae_constraint} \\
    & \bm x(0) = \bm x_0 \in \ker(C),
\end{align}
\end{subequations}
where $L \in \mathbb{C}^{n \times n}$ is the generator of the unconstrained dynamics, $C \in \mathbb{C}^{m \times n}$ is the constraint operator, and $\lambda(t) \in \mathbb{C}^m$ is the vector of Lagrange multipliers. We assume that $C$ has full row rank, such that $CC^{\dagger}$ is invertible.

While not every linear DAE is initially presented in the form of \cref{eq:original_dae}, various techniques exist to transform a general system into this structure. For high-index DAEs, one may perform index reduction to express it as an index~1 DAE. We refer readers to \cite{wanner1996solving} for more details.

\begin{lemma}
Suppose a linear DAE can be written in the semi-explicit form
\begin{equation}
\bm x' = L\bm x + G\bm \mu,\qquad C\bm x=0,
\end{equation}
where \(C\) has full row rank and $CG$ is invertible. If
\begin{equation}
\Ran(G)=\Ran(C^\dagger),
\end{equation}
then, after a nonsingular change of multiplier variable, the DAE can be written as \cref{eq:original_dae}.
\end{lemma}

\begin{lemma} \label{lem: index 1 to target}
    For a  general linear index-one DAE, after row reduction one may obtain
    \begin{equation}
       \mat{I & 0 \\ 0 & 0}\mat{\bm x_1 \\ \bm x_2}' = \mat{A_{11} & A_{12} \\ A_{21} & A_{22}} \mat{\bm x_1 \\ \bm x_2}.
    \end{equation}
    If $A_{22}$ is nonsingular, by defining 
    \begin{equation} L = \mat{ A_{11} & A_{12} \\ -A_{22}^{-1}A_{21}A_{11} & -A_{22}^{-1}A_{21}A_{12}}, \qquad C = \mat{A_{21} & A_{22}},\qquad \bm x = \mat{\bm x_1 \\ \bm x_2},\end{equation}
    and the Lagrange multiplier $\bm\lambda$ is chosen by,
    \begin{equation} \bm \lambda = -(CC^\dag)^{-1}CL\bm x,\end{equation}
    we obtain the form in \cref{eq:dae_ode,eq:dae_constraint}.
\end{lemma}
\begin{remark}
With the choice of $C$ and $L$ in \cref{lem: index 1 to target}, one has $CL=0$. Hence the multiplier in the constrained form satisfies $\bm\lambda=0$. Thus, \cref{lem: index 1 to target} should be interpreted as rewriting the index-one DAE through its Schur-complement reduced dynamics on the constraint manifold, rather than producing a nontrivial Lagrange-multiplier representation.
\end{remark}

To construct a unitary dilation, we decompose the operator $L$ into its Hermitian and anti-Hermitian parts:
\begin{equation}
    L = -iH + K, \quad \text{where} \quad H = \frac{i(L - L^{\dagger})}{2}, \quad K = \frac{L + L^{\dagger}}{2}.
\end{equation}
Note that both $H$ and $K$ are Hermitian. We define an ancillary Hilbert space $\mathcal{H}_{\mathcal{A}}$ (the ancilla space) equipped with an operator $F$ and two states $\ket{r}, \ket{l} \in \mathcal{H}_{\mathcal{A}}$ that satisfy the moment-matching dilation property: $\bra{l} F^k \ket{r} = 1$ for all $k \ge 0$. The moment-matching dilation has been proposed in \cite{li2025linear}.

We define the dilated Hamiltonian $\tH$ acting on $\mathcal{H}_{\mathcal{A}} \otimes \mathbb{C}^n$ as:
\begin{equation} \label{eq: tH}
    \tH := I_{\mathcal{A}} \otimes H + i F \otimes K.
\end{equation}
We assume that $F^\dagger=-F$, so that $iF$ is Hermitian. Since $H$ and $K$ are Hermitian,
the operator $\tH$
is Hermitian.

\begin{theorem}[Constrained Moment-Dilation] \label{thm:dilation}
Consider the DAE system \cref{eq:original_dae}. Define the lifted constraint operator $D \coloneqq I_{\mathcal{A}} \otimes C$. Let the dilated state $\Psi(t) \in \mathcal{H}_{\mathcal{A}} \otimes \mathbb{C}^n$ satisfy:
\begin{subequations} \label{eq: dilated_system}
\begin{align}
    \Psi'(t) &= -i \tH \Psi(t) + D^{\dagger} \Lambda(t), \\
    D \Psi (t) &= 0, \\
    \Psi(0) &= \ket{r} \otimes \ket{ \bm  x_0},
\end{align}
\end{subequations}
where $\Lambda(t) := -(D D^{\dagger})^{-1} D (-i \tH) \Psi(t)$. Then, the projection onto the physical space exactly recovers the solution to the original DAE \cref{eq:original_dae}:
\begin{equation}
    (\bra{l} \otimes I) \Psi(t) = \bm x(t), \quad \forall t \ge 0.
\end{equation}
\end{theorem}

\begin{proof}
To prove that $(\bra{l} \otimes I) \Psi(t) = x(t)$, we demonstrate that the projected dilated state satisfies the same evolution and initial conditions as the original DAE.

The original DAE is given by $\bm x' = L\bm x + C^\dagger \lambda$ with $C\bm x=0$. Differentiating the constraint $C\bm x(t)=0$ yields $C\bm x'(t) = 0$. Since $C$ has full row rank, $CC^\dagger$ is invertible. Solving the evolution equation, we have,
\begin{equation}
    C(L\bm x + C^\dagger \lambda) = 0 \implies \lambda = -(CC^\dagger)^{-1}CL\bm x.
\end{equation}
Substituting $\lambda$ back into the ODE, we obtain the projected evolution:
\begin{equation} \label{eq: reduced_piL}
\bm x' = \left( I - C^\dagger(CC^\dagger)^{-1}C \right) L\bm x = \Pi L \bm x, \end{equation}
where $\Pi \coloneqq I - C^\dagger(CC^\dagger)^{-1}C$ is the orthogonal projector onto $\ker(C)$. The solution is 
\begin{equation} \bm x(t) = e^{\Pi L t} \bm x_0. \end{equation}
The dilated system is $\Psi'(t) = -i \tH \Psi(t) + D^\dagger \Lambda(t)$ with $D\Psi(t) = 0$. Note that $DD^\dagger = I_A \otimes (CC^\dagger)$ is invertible.
By the same logic as above, the effective evolution on the dilated space is:
\begin{equation} \Psi'(t) = P (-i \tH) \Psi(t), \quad P \coloneqq \IAd\otimes I - D^\dag (D D^\dag)^{-1}D. \end{equation}
The solution is $\Psi(t) = e^{P (-i \tH) t} (\ket{r} \otimes x_0)$.
Furthermore, substituting $D= \IAd \otimes C$ and using the properties of tensor product, we have,
\begin{equation}
    \begin{aligned}
    P &= (\IAd \otimes I) - (\IAd \otimes C)^\dagger \left( (\IAd \otimes C)(\IAd \otimes C)^\dagger \right)^{-1} (\IAd \otimes C) \\
    &= (\IAd \otimes I) - (\IAd \otimes C^\dagger) (\IAd \otimes CC^\dagger)^{-1} (\IAd \otimes C) \\
    &= (\IAd \otimes I) - (\IAd \otimes C^\dagger) (\IAd \otimes (CC^\dagger)^{-1}) (\IAd \otimes C) \\
    &= \IAd \otimes I - \IAd \otimes (C^\dagger (CC^\dagger)^{-1} C) \\
    &= \IAd \otimes (I - C^\dagger (CC^\dagger)^{-1} C) = \IAd \otimes \Pi.    
\end{aligned}
\end{equation}

Let $\bm y(t) = (\bra{l} \otimes I) \Psi(t)$. Expanding the exponential propagator:
\begin{equation}\bm  y(t) = (\bra{l} \otimes I) \sum_{k=0}^{\infty} \frac{t^k}{k!} \left[ P (-i \tH) \right]^k (\ket{r} \otimes \bm x_0). \end{equation}
By expanding the operator $P (-i \tH)$, we obtain
\begin{equation} P (-i \tH) = \IAd \otimes \Pi (\IAd\otimes (-iH) + F\otimes K ) = \IAd \otimes (-i\Pi H) + F\otimes \Pi K . \end{equation}
Then any term in the expansion of $[ P (-i \tH) ]^k (\ket{r} \otimes \bm  x_0)$ takes the form:
\begin{equation} (F^m \ket{r}) \otimes (\mathcal{A}_{k,m} \bm x_0), \end{equation}
where $\mathcal{A}_{k,m}$ is a product of $k$ operators consisting of $m$ factors of $\Pi K$ and $(k-m)$ factors of $-i \Pi H$.

Applying $(\bra{l} \otimes I)$ and utilizing the property $\bra{l} F^k \ket{r} = 1$ for all $k \ge 0$:
\begin{equation}(\bra{l} \otimes I) \left( F^m \ket{r} \otimes \mathcal{A}_{k,m} \bm x_0 \right) = \langle l | F^m | r \rangle \mathcal{A}_{k,m} \bm x_0 = \mathcal{A}_{k,m} \bm x_0.\end{equation}
This implies that the projection $(\bra{l} \otimes I)$ maps the dilated generator $P(-i\tH)$ back into the physical generator:
\begin{equation} (\bra{l} \otimes I) [P (-i \tH)]^k (\ket{r} \otimes \bm x_0) = (-i \Pi H + \Pi K)^k \bm x_0 = (\Pi L)^k \bm x_0. \end{equation}
Summing the series gives the desired equality:
\begin{equation}\bm  y(t) = (\bra{l} \otimes I) \sum_{k=0}^{\infty} \frac{t^k}{k!} \left[ P (-i \tH) \right]^k (\ket{r} \otimes \bm x_0) = \sum_{k=0}^{\infty} \frac{t^k}{k!} (\Pi L)^k \bm x_0 = e^{\Pi L t} \bm x_0 = \bm x(t). \end{equation}
\end{proof}

\begin{Prop} \label{prop: QZE}
    The DAE in \cref{eq:original_dae} can be dilated into the block form of \cref{eq: dilated_system}:
    \begin{equation}
    \mat{ I & 0 \\ 0 & 0 } \mat{ \Psi' \\ \Lambda'} = 
    -i \mat{\tH & i D^{\dagger} \\ -i D & 0 } 
    \mat{ \Psi \\ \Lambda }.
    \end{equation}
    Identifying $P := I - D^\dag(D D^\dag)^{-1} D$ as the orthogonal projector onto $\ker(D)$, the evolution simplifies to:
    \begin{equation} \label{eq:zeno_form}
        \Psi' = -i P\tH P \Psi, \quad \Psi(0) \in \text{Range}(P).
    \end{equation}
\end{Prop}
Starting from the form in \cref{eq:zeno_form}, the dynamics can be implemented by leveraging the Quantum Zeno Effect (QZE). This phenomenon is well-established within quantum mechanics and quantum information theory. We provide a more detailed pedagogical treatment in \cref{sec: QZE}.

Thus far, we have first dilated the constrained DAE and then reduced it to a projected Hamiltonian evolution. In what follows, we demonstrate that this procedure is equivalent to first reducing the original DAE to an unconstrained system via the Schur complement and then dilating the resulting operator into the larger Hilbert space. This equivalence confirms that the dilation and projection operations commute, as formalized in the following proposition and \cref{fig: dilation_commute}.

\begin{Prop}\label{prop:schur_equiv}
The constrained DAE \cref{eq:original_dae} can be simplified to the reduced evolution
\begin{equation}\label{eq: reduced_schur}
    \bm x'(t)=\Pi L\Pi \bm x(t),\qquad \bm x(0)=\bm x_0\in\Ran(\Pi).
\end{equation}
Moreover, the dilated constrained system \cref{eq:zeno_form} is the moment-matching dilation of the reduced operator $\Pi L\Pi$ on the physical subspace.
\end{Prop}

\begin{proof}
Recall \cref{eq: reduced_piL} in the proof of \cref{thm:dilation},
\begin{equation}
\bm x'(t)=\bigl(I-C^\dagger(CC^\dagger)^{-1}C\bigr)L\bm x(t)=\Pi L\bm x(t).
\end{equation}
Since $\bm x(t)\in\ker(C)=\Ran(\Pi)$ for all $t$, we also have $\Pi x(t)=x(t)$, hence
\begin{equation}
\bm x'(t)=\Pi L\Pi \bm x(t).
\end{equation}
Now recall that
\begin{equation}
\widehat H = I_{\mathcal A}\otimes H + iF\otimes K,
\qquad
L=-iH+K.
\end{equation}
Because $P=I_{\mathcal A}\otimes \Pi$, we have
\begin{equation}
P(-i\widehat H)P
=
(\IAd\otimes \Pi)\bigl(\IAd\otimes (-iH)+F\otimes K\bigr)(I_{\mathcal A}\otimes \Pi) =
\IAd\otimes (-i\Pi H\Pi)+F\otimes \Pi K\Pi.
\end{equation}
Hence the expansion of $(P(-i\widehat H)P)^k(|r\rangle\otimes \bm x_0)$ is a sum of terms of the form $F^m|r\rangle\otimes \mathcal B_{k,m} \bm x_0,$ where $\mathcal B_{k,m}$ is a product consisting of $m$ factors of $\Pi K\Pi$
and $k-m$ factors of $-i\Pi H\Pi$.
 By applying $(\langle l|\otimes I)$ and using the moment condition
$\langle l|F^m|r\rangle =1$ for any nonnegative integer $m$, it
gives
\begin{equation}
(\langle l|\otimes I)\,(P(-i\widehat H)P)^k\,(|r\rangle\otimes \bm x_0) = (-i\Pi H\Pi+\Pi K\Pi)^k \bm x_0
= (\Pi L \Pi)^k \bm x_0
\end{equation}
Therefore, under a same dilation framework, we would obtain
\begin{equation} (\bra{l}\otimes I)e^{-it P\tH P}\ket{r}\otimes \bm x_0 = e^{t\, \Pi L\Pi}\bm x_0. \end{equation}
\end{proof}

Under a fixed dilation framework, we may denote the dilation map $\mathfrak{D}: \ker(C)\to \Anc\otimes\ker(C)$ and the map back to physical Hilbert space be $\mathfrak{R}: \Anc\otimes \ker(C)\to \ker(C)$. This setup yields the commutative relationship illustrated in \cref{fig: dilation_commute}.

\begin{figure}[ht]
\begin{equation}
\begin{tikzcd}[column sep=huge, row sep=huge]
\left\{
\begin{array}{l}
\bm x'(t)=L\bm x(t)+C^\dagger \bm \lambda(t),\\
C\bm x(t)=0
\end{array}
\right.
\arrow[r, "\text{Schur reduction}"]
\arrow[d, shift left=0.9ex, "\mathfrak D"]
&
x'(t)=\Pi L\Pi x(t)
\arrow[d, shift left=0.9ex, "\mathfrak D"]
\\
\left\{
\begin{array}{l}
\Psi'(t)=-i\widehat H\Psi(t)+D^\dagger \Lambda(t),\\
D\Psi(t)=0
\end{array}
\right.
\arrow[u, shift left=0.9ex, "\mathfrak R"]
\arrow[r, "\text{Schur reduction}"']
&
i\Psi'(t)=P\widehat H P\,\Psi(t)
\arrow[u, shift left=0.9ex, "\mathfrak R"]
\end{tikzcd}
\end{equation}
\caption{Equivalence between the dilation of the Schur complement system and the full DAE. $\mathfrak D$ refers to the dilation mapping while $\mathfrak R$ indicates the opposite, recovery  map.  }
    \label{fig: dilation_commute}
\end{figure}

We now provide a specific finite-dimensional truncated moment dilation choice from \cite{li2025linear}. It is worthwhile to mention that there exists other dilation choices, we refer readers to \cite{li2025linear} for more examples. Suppose the discrete ancilla Hilbert space $\Anc$ has dimension $M+1$. For the discrete SBP approximation on a uniform grid $p_j=j\delta $, $\delta =\frac{1}{M}$, with trapezoidal weights 
\begin{equation}
W=\operatorname{diag}(w_0,\dots,w_M),\qquad
w_0=w_M=\frac{\delta }{2},\quad w_j=\delta  \ \ (1\le j\le M-1),
\end{equation}
then define
\begin{equation}
F_{\delta }=
\mat{
0 & \frac{1}{4\sqrt2} & 0 & 0 & \cdots & 0 \\
-\frac{1}{4\sqrt2} & 0 & \frac34 & 0 & \cdots & 0 \\
0 & -\frac34 & 0 & \frac54 & \cdots & 0 \\
0 & 0 & -\frac54 & 0 & \ddots & \vdots \\
\vdots & \vdots & \vdots & \ddots & \ddots & \frac{2M-1}{4} \\
0 & 0 & 0 & \cdots & -\frac{2M-1}{4} & 0
}.
\end{equation}
Thus for \(\theta=\frac12\), the discrete ancilla operator used in the dilation is $\theta F_{\delta }=\frac12\,F_{\delta }$. Let $\beta=\frac1\theta-\frac12=\frac32.$ Then we take
\begin{equation}
\ket{r_{\delta }}
=
Z^{-1}\sum_{j=0}^M p_j^{3/2}\sqrt{w_j}\,\ket{j},
\qquad
Z=\left(\sum_{j=0}^M w_j p_j^3\right)^{1/2}.
\end{equation}
For $0< j*<M$ for some $j^*$, we pick
\begin{equation} \bra{l_{\delta }} = \frac{1}{\bra{j^*}\ket{r_{\delta }}}\bra{j^\ast}.
\end{equation}
The moment-matching pair is then $\left(\left(\bra{l_{\delta }}, \ket{r_{\delta }}, \frac{1}{2}F_{\delta }\right)\right)$, which satisfies exact moments up to order $M-j_\ast-1.$

\subsection{Quantum Zeno Form and Projective Dynamics}
\label{sec: QZE}

We now explain the connection between the projected Hamiltonian formulation derived above and
quantum Zeno dynamics. The quantum Zeno effect refers to the suppression of transitions out of a
subspace by repeated measurements or projections. In its dynamical form, if $P$ is an orthogonal
projector and $H$ is Hermitian, then, in finite dimension,
\begin{equation}
\label{eq:Zeno_limit}
    \lim_{N\to\infty}\Bigl(Pe^{-it\tH/N}P\Bigr)^N
    =
    e^{-itP\tH P}P .
\end{equation}
Thus repeated restriction to $\Ran(P)$ produces an effective Hamiltonian evolution generated by
the compressed Hamiltonian $P\tH P$ on the projected subspace; see, e.g.,
\cite{dizaji2024hamiltonian,facchi2008quantum,hahn2022unification}. The same confinement mechanism has also been realized and studied through repeated or weak
measurement protocols in several experimental and control settings
\cite{Schaefer2014,Kalb2016,Harrington2017,Sorensen2018}.

This is precisely the structure obtained from the constrained DAE after dilation. The original
system
\begin{equation}
    x'(t)=Lx(t)+C^\dagger\lambda(t),\qquad Cx(t)=0,
\end{equation}
evolves only on the constraint subspace $\ker(C)$. After dilation, the constraint becomes
\begin{equation}
    D\Psi(t)=0,\qquad D=I_{\mathcal A}\otimes C,
\end{equation}
and the orthogonal projector onto the lifted admissible subspace is
\begin{equation}
    P=I-D^\dagger(DD^\dagger)^{-1}D.
\end{equation}
By \cref{prop: QZE}, the dilated constrained dynamics are equivalently
\begin{equation}
\label{eq:projected_dynamics}
    i\Psi'(t)=P\tH P\,\Psi(t),
    \qquad
    \Psi(0)\in\Ran(P).
\end{equation}
Since $\tH $ is Hermitian, the operator $P\tH P$ is Hermitian on
$\Ran(P)$, and therefore \cref{eq:projected_dynamics} defines a unitary evolution on the
constraint subspace:
\begin{equation}
    \Psi(t)=e^{-itP\tH P}\Psi(0).
\end{equation}
Moreover, the constraint is automatically preserved. Indeed,
\begin{equation}
    (I-P)\Psi'(t)= -i(I-P)P\tH P\Psi(t)=0,
\end{equation}
so if $\Psi(0)\in\Ran(P)$, then $\Psi(t)\in\Ran(P)=\ker(D)$ for all $t\ge0$.

The Zeno viewpoint gives two complementary interpretations. In the compressed-Hamiltonian
viewpoint, one directly simulates the effective Hamiltonian $P\tH  P$. In the
repeated-projection viewpoint, the same dynamics arise formally as
\begin{equation}
\label{eq:Zeno_dilated}
    \Psi(t)
    =
    \lim_{N\to\infty}
    \Bigl(Pe^{-it\tH /N}P\Bigr)^N\Psi(0).
\end{equation}
Thus the role of the Lagrange multiplier in the original DAE is replaced, after dilation, by a
projection mechanism that continuously confines the state to the admissible subspace. The physical
solution is then recovered by the moment projection
\begin{equation}
    x(t)=(\langle l|\otimes I)\Psi(t).
\end{equation}

This connection is useful both conceptually and algorithmically. Conceptually, it identifies
constrained DAE dynamics with quantum Zeno-type projected evolution. Algorithmically, it suggests
two possible routes: one may approximate the repeated-projection formula
\eqref{eq:Zeno_dilated}, or one may coherently construct a block-encoding of the compressed
Hamiltonian $P\tH  P$ and apply Hamiltonian simulation. In the remainder of this paper, we
focus on the latter, fully coherent block-encoding approach.

\section{Implementation of the Quantum Algorithms}

The dilation developed in \cref{section: math_form} leads to the projected Schr\"odinger dynamics \cref{eq:zeno_form}. We now discuss how $P$ and $\tH$ may be implemented on a quantum computer efficiently.

If \(P\widetilde H P\) admits an explicit spectral representation, this structure should be used
directly. For periodic, constant-coefficient problems, a quantum Fourier transform diagonalizes the
discrete Laplacian and the Leray projector acts mode by mode, reducing the dynamics to spectral
multipliers. This is closely related to the quantum spectral framework for PDEs proposed in
\cite{HuangAntonioliBarbaresco2026}. Here we focus on the more general constrained DAE setting,
motivated by Stokes discretizations on domains and with boundary conditions for which such global
Fourier diagonalization is not available.

\paragraph{Block-encoding of the lifted constraint.}
Assume that we have oracle access, or an efficient block-encoding, of the constraint matrix $C$.
Since $D=\IAd\otimes C$, a block-encoding of $D$ is obtained immediately by tensoring the
block-encoding of $C$ with the identity on the ancilla register. Thus, if $U_C$ is an
$(\alpha_C,a,\varepsilon_C)$ block-encoding of $C$, then
\begin{equation}
U_D := I_A\otimes U_C
\end{equation}
is an $(\alpha_C,a,\varepsilon_C)$ block-encoding of $D$.

Because $P$ is the orthogonal projector onto $\ker(D)$, it can be viewed spectrally as the
indicator of the singular value $0$ of $D$. More precisely, if
\begin{equation}
D = \sum_j \sigma_j \, |u_j\rangle \langle v_j|,
\end{equation}
then
\begin{equation}
P = \sum_{\sigma_j=0} |v_j\rangle\langle v_j|
    = I - D^\dagger(DD^\dagger)^{-1}D.
\end{equation}
Hence the implementation of $P$ reduces to separating the zero singular-value subspace from the
nonzero singular-value subspace of $D$.

\paragraph{Implementation of $P$ via QSVT.}
A natural route is to use quantum singular value transformation (QSVT) on a block-encoding of
$D$ \cite{gilyen2019quantum}. After rescaling, we may assume that the singular values of $D/\alpha_C$ lie in $[0,1]$.
Let the nonzero singular values of $D$ are bounded below by a gap $\gamma>0$,
namely
\begin{equation}
\sigma_j(D)\in \{0\}\cup [\gamma,\alpha_C].
\end{equation}
The key idea is to choose a polynomial $p(\cdot)$ that approximates the step function,
\begin{equation} \label{eq: pulse}
p(x)\approx
\begin{cases}
1, & x=0,\\
0, & x\in [\kappa^{-1},1].
\end{cases}
\end{equation}
Here $\kappa \coloneqq \alpha_C/\gamma$ is the ratio between the largest and smallest singular values, and therefore is interpreted as the condition number of $D$. Applying QSVT to a block-encoding of $D/\alpha_C$ implements a singular-value transformation
which approximates the right singular-vector projector $P.$
Using polynomial approximation properties for the pulse function \cref{eq: pulse}, the degree of the polynomial $p$ can be estimated in terms of the precision $\varepsilon$ and the condition number \cite{gilyen2019quantum}.  

\paragraph{Implementation of the dilated Hamiltonian.}
We next turn to the Hermitian operator
\begin{equation}
\tH = I_A\otimes H + iF\otimes K.
\end{equation}
Assume that the original generator $L$ admits an efficient block-encoding. Since
\begin{equation}
H=\frac{i(L-L^\dagger)}{2},
\qquad
K=\frac{L+L^\dagger}{2},
\end{equation}
block-encodings of $H$ and $K$ follow from linear combinations of the block-encodings of $L$ and
$L^\dagger$. Thus, the non-Hermitian generator is separated into a skew-Hermitian part, encoded
by $H$, and a non-Hermitian  part, encoded by $K$, exactly as required by the dilation
construction.

Now, suppose that $U_H$ and $U_K$ are block-encodings of $H$ and $K$, respectively. Because the ancilla operator $F$ is sparse and highly structured, it is straightforward to block-encode; let $U_F$ denote this block-encoding. The term $\IAd \otimes H$ is implemented by the tensor product of $U_H$ and $\IAd$, while $iF \otimes K$ is realized via the tensor product of $U_F$ and $U_K$ together with the scalar phase factor $i$. The total operator $\tH$ is then obtained through a standard linear-combination of unitaries. Consequently, once $L$ and $F$ are efficiently accessible, so is $\tH$.

\paragraph{Projected Hamiltonian simulation.}
After constructing block-encodings of both $P$ and $\tH$, one may proceed in either of
the two ways already suggested by the quantum Zeno form. First, one may directly block-encode $P\tH P$ by composing the block encodes of $P$ and $\tH$, and then simulate the unitary evolution $e^{-itP\tH P}.$
Since $P$ and $\tH$ are Hermitian, the operator $P\tH P$ is Hermitian on the full
space and acts as the effective Hamiltonian on the invariant subspace $\Ran(P)$.

Second, one may use the Zeno product formula \cref{eq:Zeno_limit}, which realizes the constrained dynamics by repeated short-time evolution under $\tH$
interlaced with projection back to the constraint subspace.

\begin{theorem}[Complexity of simulating projected Hamiltonian dynamics]
\label{thm:end_to_end_complexity}
Let $\tH$ be Hermitian, let $D$ be a constraint operator, and let
\begin{equation}
P = I - D^\dagger (DD^\dagger)^{-1} D
\end{equation}
denote the orthogonal projector onto $\ker(D)$. Consider the projected Schr\"odinger equation
\begin{equation}
i\Psi'(t)=P\tH P\,\Psi(t), \qquad \Psi(0)=\Psi_0,
\end{equation}
with $\Psi_0\in\Ran(P)$. Assume that
\begin{enumerate}[label=(\alph*)]
    \item there is an $(\alpha_{\tH},a_{\tH},\eta_{\tH})$ block-encoding of $\tH$;
    \item there is an $(\alpha_D,a_D,\eta_D)$ block-encoding of $D$;
    \item the nonzero singular values of $D$ satisfy the gap condition
    \begin{equation}
    \sigma(D)\subset \{0\}\cup[\gamma,\alpha_D]
    \qquad \text{for some }\gamma>0;
    \end{equation}
    \item one use of the block-encodings of $\tH$ and $D$ has gate cost
    $T_{\tH}$ and $T_D$, respectively.
\end{enumerate}
Let $t> 0$ and $\varepsilon\in(0,1)$. If the block-encoding errors are chosen so that
\begin{equation}
\eta_{\tH}=O\!\left(\frac{\varepsilon}{t}\right),
\qquad
\eta_P=O\!\left(\frac{\varepsilon}{\alpha_{\tH} t}\right),
\end{equation}
where $\eta_P$ is the target error in the block-encoding of $P$, then there exists a quantum
algorithm that outputs a state $\ket{\Phi}$ satisfying
\begin{equation}
\bigl\|\ket{\Phi}-e^{-itP\tH P}\ket{\Psi_0}\bigr\|\le \varepsilon,
\end{equation}
with query complexity to the block-encodings 
\begin{equation}
O\!\left( \left( \alpha_{\tH} t + \frac{\log(1/\varepsilon)}
{\log\!\left(e+\frac{\log(1/\varepsilon)}{\alpha_{\tH} t}\right)} \right) (1+\mathfrak p) \right),
\end{equation}
where
\begin{equation}
\mathfrak p = O\!\left( \frac{\alpha_D}{\gamma} \log\!\frac{1}{\eta_P} \right) = O\!\left( \kappa(D) \log\!\frac{\alpha_{\tH}t}{\varepsilon} \right),
\end{equation}
and the block-encoding of $D$ is chosen sufficiently accurate: $\eta_D= O(\eta_p/\mathfrak p)$. 

The corresponding gate complexity is
\begin{equation} O\!\left( \left( \alpha_{\tH} t + \frac{\log(1/\varepsilon)}
{\log\!\left(e+\frac{\log(1/\varepsilon)}{\alpha_{\tH} t}\right)} \right) \bigl(T_{\tH}+\mathfrak p\,T_D\bigr) \right),
\end{equation}
up to standard ancilla overheads.
\end{theorem}

\begin{proof}
By the spectral-gap assumption on $D$, QSVT yields an $\eta_P$-approximate block-encoding of the
orthogonal projector $P$ onto $\ker(D)$ using a polynomial of degree
\begin{equation}
\mathfrak p = O\!\left( \frac{\alpha_D}{\gamma}\log\!\frac{1}{\eta_P} \right), \end{equation}
and hence using $O(\mathfrak p)$ queries to the block-encoding of $D$ and $D^\dagger$.

Let $U_P$ be the resulting $(1,a_P,\eta_P)$ block-encoding of $P$, and let
$U_{\tH}$ be the given $(\alpha_{\tH},a_{\tH},\eta_{\tH})$
block-encoding of $\tH$. By two applications of the product rule for block-encodings,
the composition $U_P\,U_{\tH}\,U_P $
is an $(\alpha_{\tH},a,\eta_{P\tH P})$ block-encoding of $P\tH P$, where $\eta_{P\tH P}
=O\!\left(\eta_{\tH}+\alpha_{\tH}\eta_P\right).$ Choosing
\begin{equation}
\eta_{\tH}=O\!\left(\frac{\varepsilon}{t}\right),
\qquad
\eta_P=O\!\left(\frac{\varepsilon}{\alpha_{\tH}t}\right),
\end{equation}
gives
\begin{equation}
\eta_{P\tH P}=O\!\left(\frac{\varepsilon}{t}\right).
\end{equation}

Applying optimal Hamiltonian simulation to this approximate block-encoding of the Hermitian
operator $P\tH P$ yields an implementation of
\begin{equation} e^{-itP\tH P} \end{equation}
with error at most $\varepsilon$ and query complexity
\begin{equation}
O\!\left( \alpha_{\tH} t + \frac{\log(1/\varepsilon)}
{\log\!\left(e+\frac{\log(1/\varepsilon)}{\alpha_{\tH} t}\right)}
\right), \end{equation}
where each query to the block-encoding of $P\tH P$ requires one use of
$U_{\tH}$ and two uses of $U_P$. Substituting the cost of $U_P$ gives the stated query
and gate bounds.
\end{proof}

\paragraph{A linear DAE example from circuit theory.}
As a concrete linear DAE that does not arise from a semidiscretized PDE, consider an
$N$-section RLC transmission-line ladder written in modified nodal analysis; see, e.g.,
\cite{HoRuehliBrennan1975,Riaza2008}. We work in the homogeneous case $u(t)\equiv 0$.

Let $v_0(t)$ denote the source-node voltage, let
\begin{equation}
\widehat{\bm  v}(t)=(v_1(t),\dots,v_N(t))^T\in\mathbb{R}^N
\end{equation}
be the ladder-node voltages, and let
\begin{equation}
\bm i(t)=(i_1(t),\dots,i_N(t))^T\in\mathbb{R}^N
\end{equation}
be the branch currents. Define
\begin{equation}
j_s(t):=-i_s(t),
\end{equation}
so that the source-node Kirchhoff law becomes $i_1-j_s=0$.

We take uniform parameters
\begin{equation}
R=0.2,\qquad L_{\mathrm{ind}}=1,\qquad C=1,\qquad G=0.05,
\end{equation}
and introduce
\begin{equation}
e_1=(1,0,\dots,0)^T\in\mathbb{R}^N,
\qquad
K_N=
\begin{bmatrix}
-1 & 1 & 0 & \cdots & 0\\
0 & -1 & 1 & \ddots & \vdots\\
\vdots & \ddots & \ddots & \ddots & 0\\
0 & \cdots & 0 & -1 & 1\\
0 & \cdots & \cdots & 0 & -1
\end{bmatrix}\in\mathbb{R}^{N\times N}.
\end{equation}
Then the dynamics take the form
\begin{align}
\dot{\widehat{\bm v}} &= -G\,\widehat{\bm v}-K_N\,\bm i, \\
\dot{\bm i} &= \bm e_1 v_0+K_N^T\widehat{\bm  v}-R\,\bm i+\lambda_2 \bm e_1, \\
\dot v_0 &= \lambda_1, \\
\dot j_s &= -\lambda_2,
\end{align}
subject to the algebraic constraints
\begin{equation}
v_0=0,\qquad i_1-j_s=0.
\end{equation}

Define
\begin{equation}
\bm x:=
\begin{bmatrix}
v_0\\ \widehat v\\ i\\ j_s
\end{bmatrix}\in\mathbb{R}^{2N+2},
\qquad
\lambda:=
\begin{bmatrix}
\lambda_1\\ \lambda_2
\end{bmatrix}\in\mathbb{R}^2.
\end{equation}
Then the system can be written in the constrained form
\begin{equation}
\bm x'(t)=L_N \bm  x(t)+D_N^\dagger\lambda(t),\qquad D_N \bm x(t)=0,
\end{equation}
with
\begin{equation}
D_N=
\begin{bmatrix}
1 & 0_{1\times N} & 0_{1\times N} & 0\\
0 & 0_{1\times N} & e_1^T & -1
\end{bmatrix},
\end{equation}
and
\begin{equation}
L_N=
\begin{bmatrix}
0 & 0_{1\times N} & 0_{1\times N} & 0\\
0_{N\times 1} & -G I_N & -K_N & 0_{N\times 1}\\
e_1 & K_N^T & -R I_N & 0_{N\times 1}\\
0 & 0_{1\times N} & 0_{1\times N} & 0
\end{bmatrix}.
\end{equation}
Thus the constraint subspace is
\begin{equation}
\ker(D_N)=\left\{x\in\mathbb{R}^{2N+2}: v_0=0,\ i_1-j_s=0\right\},
\end{equation}
and the orthogonal projector onto the physical subspace is
\begin{equation}
P_N=I-D_N^\dagger(D_ND_N^\dagger)^{-1}D_N.
\end{equation}

This family is particularly simple from the standpoint of
\cref{thm:end_to_end_complexity}. Indeed,
\begin{equation}
D_ND_N^\dagger=
\begin{bmatrix}
1 & 0\\
0 & 2
\end{bmatrix},
\end{equation}
so the nonzero singular values of $D_N$ are exactly
\begin{equation}
\sigma(D_N)=\{1,\sqrt{2}\}.
\end{equation}
Hence
\begin{equation}
\gamma=1,\qquad \alpha_D=\sqrt{2},\qquad \frac{\alpha_D}{\gamma}=\sqrt{2}=O(1),
\end{equation}
independently of $N$. In particular, the QSVT cost of implementing the projector $P_N$
depends only logarithmically on the target precision.

The dynamical matrix $L_N$ is also uniformly bounded in $N$. Since $\|e_1\|=1$ and
$\|K_N\|\le 2$, one has
\begin{equation}
\|L_N\|
\le
\|e_1\|+\|K_N\|+\|K_N^T\|+R+G
\le
1+2+2+0.2+0.05,
\end{equation}
and therefore
\begin{equation}
\|L_N\|=O(1),
\end{equation}
again uniformly in $N$. Consequently, the corresponding dilation Hamiltonian
$\tH_N$ also satisfies
\begin{equation}
\|\tH_N\|=O(1),
\end{equation}
up to constant factors coming from the fixed ancilla operator.

Heuristically, if sparse block-encodings of $L_N$ and $D_N$ are available at $O(1)$ query
cost, then \cref{thm:end_to_end_complexity} predicts an overall simulation cost that is
essentially linear in the evolution time,
\begin{equation}
\widetilde O(t),
\end{equation}
up to polylogarithmic dependence on the target precision. In this family, the constraint
enforcement remains inexpensive because $D_N$ has constant condition number, and the leading
dependence is governed by the Hamiltonian-simulation term rather than by the projector
construction.

\section{Application to Unsteady Stokes Flow} \label{sec: Stokes}

In this section, we study the unitary dilation framework for a physically meaningful DAE arising from incompressible fluid flow, namely the Stokes equations. The Stokes system can be viewed as a simplified and linear form of the incompressible Navier--Stokes equations obtained by neglecting the nonlinear convective term. It therefore serves as a natural first model problem: it retains the essential differential-algebraic structure caused by the incompressibility constraint, while avoiding the additional analytical and numerical difficulties introduced by nonlinearity. For this reason, the Stokes equations provide a clean setting in which to formulate the problem as a dilated DAE, analyze the associated operators, and test the accuracy of the dilation approach.

The purpose of recalling these classical estimates is not to provide a new numerical analysis of
Stokes discretizations. Rather, these results establish the baseline accuracy and cost scalings
against which the quantum simulation procedure will be compared. In particular, second-order
spatial and temporal methods suggest the common accuracy balance $\Delta t\sim h$, while optimal
elliptic solvers lead to the standard classical work estimate used later.

Consider an unsteady incompressible Stokes flow for the velocity $\bm u(\bm x,t)$, the pressure field $p(x,t)$ on a domain $\Omega \subset \mathbb{R}^d$, compatible boundary conditions, and external force $\bm f:\Omega\times[0,T]\to \Omega$:
\begin{subequations} \label{eq:stokes_strong_application}
\begin{align}
    \bu_t - \Delta \bu + \nabla p &= \bm f, \\
    \nabla \cdot \bu &= 0.
\end{align}
\end{subequations}
For non-homogeneous, but time-independent boundary data satisfying the compatibility condition
\begin{equation}
\int_{\partial\Omega} g\cdot n\,ds=0,
\end{equation}
one may introduce a divergence-free lifting $\bm y$ satisfying
\begin{equation}
-\Delta \bm y+\nabla p_y=0,\qquad
\nabla\cdot \bm y=0,\qquad
\bm y|_{\partial\Omega}=g.
\end{equation}
Then $\overline{\bm u}=\bm u-\bm y$ satisfies a homogeneous Dirichlet Stokes system, with a
modified pressure and forcing.

\subsection{Regularity and Function Spaces}

Let $\Omega \subset \mathbb{R}^d$ ($d=2,3$) be a bounded domain with boundary $\partial\Omega$, and let $T>0$ be a final time. We use standard Sobolev space notation throughout. For $m \in \mathbb{R}$, the Sobolev space $\H^m(\Omega)$ is equipped with norm $\norm{\cdot}_{m}$. In particular, $\H^0(\Omega)=L^2(\Omega)$, 
with norm $\norm{\cdot}$ and inner product
\begin{equation}
(f,g)=\int_\Omega f(x)g(x)\,dx.
\end{equation}
We denote the Sobolev space of $\H^1$ functions with vanishing trace on the boundary by
\begin{equation}
\H_0^1(\Omega)=\{v\in \H^1(\Omega): v|_{\partial\Omega}=0\}.
\end{equation}
Its dual space is denoted by $\H^{-1}(\Omega)$.
In our later assumption, we introduce Sobolev $\H^4$, i.e.:
\begin{equation} \label{H4}
\H^4(\Omega) = \{\bm u\in L^2(\Omega): D^\alpha\bm u\in L^2(\Omega) \text{ for every multi-index } \alpha \text{ with }|\alpha|\le 4\}. 
\end{equation}
For the pressure variable, we use the zero-mean space
\begin{equation}
L_0^2(\Omega)=\left\{q\in L^2(\Omega): \int_\Omega q(x)\,dx=0\right\}.
\end{equation}
A function $f$ belongs to $L^2(0,T;X)$ if
\begin{equation}
\int_0^T \norm{f(t)}_X^2\,dt < \infty,
\end{equation}
In other words, $L^2(0,T;X)$ denotes square-integrable functions of time taking values in the spatial function space $X$.
Furthermore, $f\in H^1(0,T;X)$ if both $f$ and its time derivative $\partial_t f$ belong to $L^2(0,T;X)$. Their norms are
\begin{equation}
\|f\|_{L^2(0,T;X)}^2 := \int_0^T \|f(t)\|_X^2\,dt,
\end{equation}
and
\begin{equation}
\|f\|_{H^1(0,T;X)}^2
:=
\int_0^T \|f(t)\|_X^2\,dt
+
\int_0^T \|\partial_t f(t)\|_X^2\,dt.
\end{equation}

\subsection{Convergence Results for Finite Difference and Finite Element Methods}

For classical time-dependent incompressible flow discretizations, including MAC schemes,
projection methods, and finite element methods, we refer to
\cite{HarlowWelch1965,Chorin1968,heywood1982finite,heywood1990finite,
temam2024navier,GiraultRaviart1986,GuermondMinevShen2006}.
In this section, we provide previous results of numerical analysis for unsteady Stokes, in both space and time discretization. Both spacial and time discretization for incompressible Stokes and Navier--Stokes have been well studied. We recommend authors to read \cite{heywood1982finite, heywood1990finite, hou2001error, chen2016finite, guermond2006overview, temam2024navier, han1998new} for further details. We first introduce a result from \cite[Theorem 3.1]{hou2001error}, which only assume low regularity on solution and external force $\bm f$. 
\begin{theorem}[Convergence rate under low regularity]
Assume that the solution $(\bm u, p)\in U\times P$ with $U$ and $P$ defined as
\begin{equation} U = L^2(0,T; \H^3(\Omega))\cap \H^1(0,T; \H^1(\Omega)), \qquad P =  L^2(0,T; \H^2(\Omega)), \end{equation}
and let $(\bm u^h, p^h)\in\H^1(0,T; \X^h)\times L^2(0,T; \M^h)$, where $\X^h\subset \H_0^1$ and $\M^h\subset L_0^2(\Omega)$ are the finite element spaces satisfying the inf-sup condition, then 
\begin{equation}
    \norm{\bm u(t) - \bm u^h(t)} + \norm{\bm u - \bm u^h}_{L^2(0,T, \H_0^1(\Omega))} \le C h^2\norm{\bm u}_{L^2(0,T; \H^{3}(\Omega))}.
\end{equation}
\end{theorem}

A popular numerical approach for Stokes equations is Marker-and-Cell (MAC) method, a finite different method. We now provide a MAC staggered grid convergence result from \cite{chen2016finite}. Since the role of the MAC scheme in our later discussion is to provide a compatible spatial
discretization, we quote a representative steady Stokes estimate for the spatial error. 
\begin{theorem}[MAC spatial convergence]
Assume that the exact solution $(\bm u,p)$ of the steady Stokes problem satisfies
\begin{equation}
\bm u \in H_0^1(\Omega)^2\cap W^{3,\infty}(\Omega)^2,
\qquad
\Delta \bm u \in H^2(\Omega),
\qquad
p \in L_0^2(\Omega)\cap W^{3,\infty}(\Omega).
\end{equation}
Here $W^{3,\infty}(\Omega)\coloneqq \{\bm u: \Omega\to\R \ | \ D^\alpha \bm u\in L^\infty(\Omega) \text{ for all }|\alpha|\le3\}.$
Let $(\bm u_h,p_h)$ be the MAC approximation on a uniform staggered grid, and let
$(\bm u_I,p_I)$ be the corresponding interpolant of $(\bm u,p)$ onto the grid unknowns.

Then, for the MAC scheme with quadratic boundary extrapolation,
\begin{equation}
|\bm u_I-\bm u_h|_{1}+\|p_I-p_h\|_{L^2(\Omega)} \le C h^2.
\end{equation}
Here the constant $C>0$ is independent of $h$ and depends only on the solution norms $\|\bm u\|_{W^{3,\infty}(\Omega)}$, $\|\Delta \bm u\|_{H^2(\Omega)}$ and $\|p\|_{W^{3,\infty}(\Omega)}.$
\end{theorem}

We provide a numerical analysis establishing second-order convergence in time for the Crank–Nicolson method, following \cite[Theorem 4.1]{heywood1990finite}. 
shows the second-order temporal accuracy of Crank--Nicolson.
\begin{theorem}[Second-order convergence in time for Crank--Nicolson]
Let $(\bm u_h(t),p_h(t))$ be the semidiscrete solution of the spatially discretized unsteady
Stokes system on $[0,T]$, and let $\{(\bm U_h^n,P_h^n)\}_{n=0}^N$ be the Crank--Nicolson
approximation with time step $k=T/N$. We denote
\begin{equation} \V = \{\bm\varphi\in \H_0^1(\Omega): \nabla\cdot \bm\varphi=0\}. \end{equation}
Assume the initial velocity and external forcing satisfy 
\begin{equation}\bm u_0 \in \H^2(\Omega)^d \cap \V, \end{equation}
\begin{equation} \bm f,\ \partial_t \bm f,\ \partial_{tt}\bm f
\in L^\infty(0,T;H^{-1}(\Omega)^d). \end{equation}
Assume that the corresponding strong solution exists on $[0,T]$ and
\begin{equation}
\sup_{t\in[0,T]}\|\nabla \bm u(t)\|_{L^2(\Omega)} \le M .
\end{equation}

Then, for sufficiently small $k$, the Crank--Nicolson velocity approximation is second order
in time in the sense that
\begin{equation}
\|\bm U_h^n-\bm u_h(t_n)\|_{L^2(\Omega)}
\le C_T\, \tau_n^{-1} k^2,
\qquad
\tau_n:=\min\{1,t_n\},
\end{equation}
for all $0<t_n\le T$. Furthermore, the associated pressure approximation satisfies
\begin{equation}
\|P_h^n-\overline p_h(t_n)\|_{L^2(\Omega)/N_h}
\le C_T\, \tau_n^{-3/2} k,
\end{equation}
where
\begin{equation}
\overline p_h(t_n):=\frac12\bigl(p_h(t_n)+p_h(t_{n-1})\bigr).
\end{equation}
\end{theorem}

Both discretization frameworks above, mixed FEM and MAC staggered grid, lead to a finite-dimensional constrained DAE system. More precisely, if $\bm x(t)$ denotes the vector of discrete velocity unknowns and $\bm p(t)$ the vector of discrete pressure unknowns, then the semidiscrete Stokes equations take the form
\begin{equation} \label{eq: numerical Stokes}
    M_v \bm x'(t) = L_h \bm x(t) + D_h^\dag \bm p(t),
    \qquad
    D_h \bm x(t)=0,
\end{equation}
where $M_v$ is the velocity mass matrix, $L_h$ is the discrete diffusion operator, and $D_h$ is the discrete divergence operator. Denote the spacial mesh size by $h$. In the case of mixed finite element discretizations, we employ mass lumping to yield a diagonal velocity mass matrix $M_{v}^{\operatorname{diag}}$, which is subsequently rescaled to the identity. For the MAC staggered-grid method, the velocity mass matrix is inherently the identity. For the autonomous Hamiltonian simulation discussion below, we focus on the homogeneous case
$\bm f_h=0$; time-dependent forcing can be treated by standard augmentation or variation-of-
constants formulations, but is not the focus here.

The natural follow-up question is how do we simulate the semidiscrete system via Hamiltonian simulation.
For the Stokes problem, after dilation, we have 
\begin{equation}\label{th4stokes}
\tH_h \coloneqq iF\otimes \Delta_h \qquad
P_h := \IAd\otimes \Pi_h.
\end{equation}
Then the reduced dilated physical
evolution is
\begin{equation}
i\frac{d}{dt}\ket{\Psi(t)}=\pH \ket{\Psi(t)},
\qquad
\pH:=P_h \tH_h P_h.
\end{equation}

A direct simulation of the reduced dynamics is still quite expensive. Indeed, for the Stokes operator one
typically has
\begin{equation}
\|\tH\|=\Theta(h^{-2}),
\end{equation}
and the cost of implementing the projector $P_h$ by QSVT introduces an additional dependence on $h$. Following from \cref{thm:end_to_end_complexity}, we have the following baseline estimate with a direct implementation of the optimal Hamiltonian simulation. 
\begin{corollary}[Direct projected simulation]\label{cor-direct-sim}
Recall that $\|\tH_h\|=\Theta(h^{-2})$. Assume furthermore that the
block-encoding cost of $P_h$ is $\widetilde O(h^{-1}).$ Then direct QSVT simulation of the dynamics $e^{-itP_h \tH_h P_h}$ has query complexity $\widetilde O(th^{-3})$.
\end{corollary}
Thus, for $t=O(1)$, the cost, in terms of the queries to the block encodings of the operators in the $D_h$ and $L_h$,  scales like $O(h^{-3})$: one factor $h^{-2}$ comes from the diffusive scale and one factor $h^{-1}$ from enforcing incompressibility.

{

\subsection{Reduced Quantum Zeno simulation by Gaussian moment dilation}
\label{subsec:gaussian_zeno_stokes}

The rough complexity estimate in \cref{cor-direct-sim} provides a useful baseline for
simulating the projected Stokes dynamics by directly block-encoding the compressed
Hamiltonian $P_h\tH_hP_h$. However, the Stokes operator has additional structure that can be leveraged to obtain a better QZE form. Specifically,  after reduction to the discrete divergence-free subspace, the homogeneous Stokes
evolution is generated by a positive operator that admits a square factorization.
This allows us to replace the generic moment dilation of $\exp(tL)$ by a more
specialized Gaussian moment dilation of $\exp(-tA^\dagger A)$.

To elaborate on the construction, we recall that
\begin{equation}
    \Pi_h := I-D_h^\dagger(D_hD_h^\dagger)^{-1}D_h
\end{equation}
is the orthogonal projector onto the discrete divergence-free velocity subspace.
For the compatible MAC or mixed finite element discretizations considered here, the
negative discrete Laplacian admits the factorization
\begin{equation}
    -\Delta_h = G_h^\dagger G_h \succeq 0,
\end{equation}
where $G_h$ is the corresponding first-order difference or finite element gradient
operator. Therefore the reduced homogeneous Stokes equation can be written as
\begin{equation}
    \frac{d}{dt}\bm u_h(t)
    =
    \Pi_h\Delta_h\Pi_h\,\bm u_h(t)
    =
    -S_h\bm u_h(t),
\end{equation}
where
\begin{equation}
    S_h:=-\Pi_h\Delta_h\Pi_h
    =
    \Pi_hG_h^\dagger G_h\Pi_h .
\end{equation}
Since $\Pi_h=\Pi_h^\dagger=\Pi_h^2$ as an orthogonal projection, we have the square factorization
\begin{equation}
\label{eq:Stokes-square-factorization}
    S_h=(G_h\Pi_h)^\dagger(G_h\Pi_h).
\end{equation}
This factorization is the main structure to be exploited in this subsection. 

\medskip 
The time evolution generated by $S_h$ has the Gaussian form $\exp(-t x^2)$ and thus we seek a dilation customized for such Gaussian dynamics.  Here and below in this subsection, the symbol $F$ denotes a new self-adjoint Gaussian
ancilla operator.  It is not the skew-Hermitian moment-dilation operator used in
\cref{section: math_form}. 

We first recall the elementary Gaussian moment identity. Let $F$ be a self-adjoint
ancillary operator on an ancillary Hilbert space $\Anc$, which is different from that in \cref{th4stokes}, and let $\ket{r},\bra{l}\in\Anc$
be a lifting-recovery pair satisfying
\begin{equation}
\label{eq:heat-moment-conditions}
    \bra{l}F^{2m+1}\ket{r}=0,
    \qquad
    \bra{l}F^{2m}\ket{r}=\frac{(2m)!}{m!},
    \qquad m=0,1,2,\ldots .
\end{equation}
Then, for any Hermitian operator $B$,
\begin{align}
    (\bra{l}\otimes I)
    e^{-iF\otimes B}
    (\ket{r}\otimes I)
    &=
    \sum_{j=0}^\infty
    \frac{(-i)^j}{j!}\bra{l}F^j\ket{r}\,B^j                                      \notag\\
    &=
    \sum_{m=0}^\infty
    \frac{(-i)^{2m}}{(2m)!}
    \frac{(2m)!}{m!}B^{2m}                                                        \notag\\
    &=
    \sum_{m=0}^\infty
    \frac{(-1)^m}{m!}B^{2m}
    =
    e^{-B^2}.
\end{align}
After replacing $B$ by $\sqrt t\,B$, we obtain
\begin{equation}
\label{eq:abstract-heat-moment-dilation}
    e^{-tB^2}
    =
    (\bra{l}\otimes I)
    e^{-i\sqrt t\,F\otimes B}
    (\ket{r}\otimes I).
\end{equation}

A concrete realization is obtained from a Gaussian ancillary state. We will leave other possible choices to future studies. Take
$\Anc=L^2(\mathbb R)$ and let $F$ be the operator
\begin{equation}
    (F\phi)(q)\coloneqq q\phi(q).
\end{equation}
In the generalized position basis,
\begin{equation}
    F\ket{q}=q\ket{q},
    \qquad
    I_{\mathcal A}=\int_{\mathbb R}\ket{q}\bra{q}\,dq .
\end{equation}
Let $\ket{g}\in L^2(\mathbb R)$ be the normalized Gaussian state with wavefunction
\begin{equation}
    g(q):=\bra{q}\ket{g}
    =
    (4\pi)^{-1/4}e^{-q^2/8}.
\end{equation}
Then
\begin{equation}
    |g(q)|^2
    =
    \frac{1}{\sqrt{4\pi}}e^{-q^2/4},
\end{equation}
which is the density of a centered Gaussian random variable with variance $2$. We choose
\begin{equation}
    |r\rangle=|g\rangle,\qquad \langle l|=\langle g|.
\end{equation}
Equivalently,
\begin{equation}
    \ket{g}=\int_{\mathbb R}g(q)\ket{q}\,dq,
    \qquad
    \bra{g}=\int_{\mathbb R}\overline{g(q)}\bra{q}\,dq .
\end{equation}
Therefore
\begin{equation}
    \bra{g}F^m\ket{g}
    =
    \int_{\mathbb R}q^m |g(q)|^2\,dq .
\end{equation}
Since $|g(q)|^2$ is an even Gaussian density,
\begin{equation}
    \bra{g}F^{2m+1}\ket{g}=0,
    \qquad
    \bra{g}F^{2m}\ket{g}
    =
    \frac{(2m)!}{m!}.
\end{equation}
Equivalently, its characteristic function satisfies
\begin{equation}
    \bra{g}e^{-iuF}\ket{g}=e^{-u^2}.
\end{equation}
Thus
\begin{equation}
\label{eq:gaussian-moment-dilation}
    e^{-tB^2}
    =
    (\bra{g}\otimes I)
    e^{-i\sqrt t\,F\otimes B}
    (\ket{g}\otimes I).
\end{equation}
Equivalently, by the Fourier representation of the Gaussian,
\begin{equation}
\label{eq:operator-gaussian-integral}
    e^{-tB^2}
    =
    \frac{1}{\sqrt{4\pi t}}
    \int_{\mathbb R}
    e^{-k^2/(4t)}e^{-ikB}\,dk .
\end{equation}
Hence the Gaussian integral is not merely an LCU representation; it is a concrete
moment-matching dilation for the function $e^{-x^2}$.

We now apply this construction to the reduced Stokes operator. Define the first-order Hermitian
dilation
\begin{equation}
    \mathscr D_h :=
    \begin{pmatrix}
        0&G_h^\dagger\\
        G_h&0
    \end{pmatrix}
\end{equation}
on the direct sum of the velocity space and the corresponding gradient space. Define the lifted
incompressibility projector
\begin{equation}
    \mathcal P_h :=
    \begin{pmatrix}
        \Pi_h&0\\
        0&I
    \end{pmatrix}.
\end{equation}
The projected first-order Hamiltonian is
\begin{equation}
    B_h:=\mathcal P_h\mathscr D_h\mathcal P_h
    =
    \begin{pmatrix}
        0&\Pi_hG_h^\dagger\\
        G_h\Pi_h&0
    \end{pmatrix}.
\end{equation}
This is again a quantum Zeno form. The incompressibility constraint is enforced by
$\mathcal P_h$, but the projected Hamiltonian is now the first-order operator
$\mathscr D_h$ rather than the second-order diffusive operator. Formally,
\begin{equation}
    e^{-ikB_h}\mathcal P_h
    =
    \lim_{r\to\infty}
    \left(
        \mathcal P_h e^{-ik\mathscr D_h/r}\mathcal P_h
    \right)^r .
\end{equation}
The square of $B_h$ is block diagonal, that is, 
\begin{equation}
\label{eq:Bh-square-new}
    B_h^2
    =
    \begin{pmatrix}
        \Pi_hG_h^\dagger G_h\Pi_h&0\\
        0&G_h\Pi_hG_h^\dagger
    \end{pmatrix}
    =
    \begin{pmatrix}
        S_h&0\\
        0&G_h\Pi_hG_h^\dagger
    \end{pmatrix}.
\end{equation}
Let $\big(\ket{0}\otimes I\big)$ denote the injection into the first block, so that
\begin{equation}
    (\ket{0}\otimes I)\bm v
    =
    \begin{pmatrix}
        \bm v\\0
    \end{pmatrix}.
\end{equation}
Then
\begin{equation}
\label{eq:Stokes-semigroup-from-Bh}
    e^{-tS_h}
    =
    (\bra{0}\otimes I)
    e^{-tB_h^2}
    (\ket{0}\otimes I).
\end{equation}
Combining \cref{eq:gaussian-moment-dilation} and
\cref{eq:Stokes-semigroup-from-Bh} gives the exact Gaussian-Zeno representation
\begin{equation}
\label{eq:exact-gaussian-zeno-Stokes}
    e^{-tS_h}
    =
    (\bra{0}\otimes I)
    (\bra{g}\otimes I)
    e^{-i\sqrt t\,F\otimes B_h}
    (\ket{g}\otimes I)
    (\ket{0}\otimes I).
\end{equation}
Equivalently,
\begin{equation}
\label{eq:Stokes-Gaussian-integral-new}
    e^{-tS_h}
    =
    (\bra{0}\otimes I)
    \left[
    \frac{1}{\sqrt{4\pi t}}
    \int_{\mathbb R}
    e^{-k^2/(4t)}e^{-ikB_h}\,dk
    \right]
    (\ket{0}\otimes I).
\end{equation}
The important point is that the unitary evolutions in the integral are generated by
$B_h$, whose norm scales like a first-order differential operator:
\begin{equation}
\label{eq:Bh-norm-scale-new}
    \|B_h\|=\|G_h\Pi_h\|\le \|G_h\|=\Theta(h^{-1}),
\end{equation}
whereas $\|S_h\|=\Theta(h^{-2})$.

This formulation is equivalent to the Gaussian-LCHS approach \cite{kharazi2026sublinear}, in which the integral is efficiently treated by Gaussian quadratures. 
\begin{equation}
\label{eq:Gaussian-LCHS-Stokes}
    e^{-tB_h^2}
    \approx
    \sum_{m=0}^{M-1}c_m(t)e^{-ik_mB_h},
    \qquad
    \sum_{m=0}^{M-1}|c_m(t)|=1+O(\varepsilon),
\end{equation}
with overall error at most $\varepsilon$. The same transmutation viewpoint also
appears in recent work on quantum simulation of non-unitary dynamics by Gaussian
superpositions of wave propagators~\cite{jin2026transmutation}. Since the
LCU normalization is constant up to the target precision, the dominant cost is the
largest Hamiltonian simulation time in the quadrature. If $B_h$ has an
$(\alpha_B,a_B,\eta_B)$ block-encoding, then the resulting query complexity is
\begin{equation}
\label{eq:Gaussian-LCHS-query-short}
    \widetilde O(\alpha_B\sqrt t)
\end{equation}
queries to the block-encoding of $B_h$, with polylogarithmic dependence on
$\varepsilon^{-1}$. In the Stokes setting, \cref{eq:Bh-norm-scale-new} gives
$\alpha_B=\Theta(h^{-1})$.

We summarize the resulting cost in the following proposition.

\begin{Prop}[Gaussian-Zeno simulation of the reduced Stokes semigroup]
\label{prop:gaussian-zeno-stokes-new}
Let
\begin{equation}
    S_h=-\Pi_h\Delta_h\Pi_h=(G_h\Pi_h)^\dagger(G_h\Pi_h)
\end{equation}
be the reduced positive Stokes operator. Assume that $G_h$ and $D_h$ admit
block-encodings with normalizations
\begin{equation}
    \alpha_G=\Theta(h^{-1}),
    \qquad
    \alpha_D=\Theta(h^{-1}),
\end{equation}
and that the nonzero singular values of $D_h$ are bounded below by a mesh-independent
constant $\gamma=\Theta(1)$. Then an $\varepsilon$-accurate block-encoding of
$e^{-tS_h}$ can be constructed with gate complexity
\begin{equation}
\label{eq:Gaussian-Zeno-general-cost-new}
    \widetilde O\!\left(
        \sqrt t\,h^{-1}\bigl(T_G+h^{-1}T_D\bigr)
    \right),
\end{equation}
where $T_G$ and $T_D$ denote the costs of one query to the block-encodings of $G_h$ and
$D_h$, respectively.
\end{Prop}

Indeed, constructing $\Pi_h$ from $D_h$ by QSVT requires
\begin{equation}
    O\!\left(
        \frac{\alpha_D}{\gamma}\log\frac{1}{\eta_\Pi}
    \right)
    =
    \widetilde O(h^{-1})
\end{equation}
queries to the block-encoding of $D_h$. Composing this projector with the
block-encoding of $G_h$ gives a block-encoding of $A_h:=G_h\Pi_h$, and the standard
Hermitian dilation of $A_h$ gives a block-encoding of $B_h$. Thus one query to
$B_h$ costs $O(T_G+h^{-1}T_D)$ up to logarithmic factors, while the Gaussian-LCHS
simulation uses $\widetilde O(h^{-1}\sqrt t)$ such queries.

Finally, we distinguish the block-encoding of the semigroup from preparation of the
time-evolved normalized state. Applying the block-encoding of $e^{-tS_h}$ to a
normalized input state $\ket{\psi_{0,h}}\in\Ran(\Pi_h)$ prepares, after postselection,
a state proportional to $e^{-tS_h}\ket{\psi_{0,h}}$. With amplitude amplification,
preparing the normalized state
\begin{equation}
    \frac{e^{-tS_h}\ket{\psi_{0,h}}}
    {\|e^{-tS_h}\ket{\psi_{0,h}}\|}
\end{equation}
introduces the additional factor
\begin{equation}
\label{eq:chi-factor-new}
    \chi_h(t,\psi_0)
    :=
    \frac{1}{\|e^{-tS_h}\ket{\psi_{0,h}}\|}.
\end{equation}
Therefore, in the generic sparse-access model,
\begin{equation}
\label{eq:Gaussian-Zeno-state-prep-cost-new}
    \mathrm{Cost}_{\mathrm{prep}}(t,h,\varepsilon)
    =
    \widetilde O\!\left(
        \chi_h(t,\psi_0)h^{-2}\sqrt t
    \right).
\end{equation}
For fixed final time and smooth initial data whose dissipative norm does not vanish as
$h\to0$, $\chi_h(t,\psi_0)$ is independent of the mesh size. In that regime, the
mesh dependence of the Gaussian-Zeno state-preparation cost is
\begin{equation}
    \widetilde O(h^{-2}\sqrt t).
\end{equation}
The improvement over the direct projected simulation comes from replacing the
second-order scale $\|S_h\|=\Theta(h^{-2})$ by the first-order square-root scale
$\|B_h\|=\Theta(h^{-1})$, while the remaining factor $h^{-1}$ in the generic model
comes from constructing the incompressibility projector.

\begin{remark}
We also considered an alternative route based on low-energy subspace Hamiltonian
simulation~\cite{low2017hamiltonian,zlokapa2024hamiltonian}. In the Stokes setting,
such a low-energy description is naturally motivated by PDE smoothness: sufficiently
regular data have most of their spectral weight in the low-frequency eigenspaces of
the discrete Stokes operator. However, in the general block-encoding model considered
here, the resulting complexity is not as favorable as the Gaussian-Zeno construction.
The main obstruction is that the incompressibility projector is still constructed from
the full-space divergence operator, whose normalization scales like
$\|D_h\|=\Theta(h^{-1})$; this projector cost offsets part of the gain obtained from
spectral localization. In contrast, the Gaussian-Zeno approach exploits the
factorization $S_h=(G_h\Pi_h)^\dagger(G_h\Pi_h)$ directly and obtains the
$\sqrt t$ dependence through the Gaussian moment dilation. It remains an interesting
open problem whether low-energy localization and Gaussian-Zeno simulation can be
combined in a complementary way, for example by constructing a filtered block-encoding
of $G_h\Pi_h$ or of the Leray projector with reduced normalization.
\end{remark}

}

\subsection{A heuristic comparison with classical cost}

We briefly compare the cost of the quantum time-evolution procedure with that of a standard
classical projection method for the semidiscrete Stokes equations. This comparison concerns only
the cost of propagating the system over a time interval of length $t$ and does not include data
loading, state preparation, measurement, or postprocessing.

\paragraph{Classical baseline.}
Assume that the spatial discretization is second order,
\begin{equation}
\|u-u_h\|\lesssim h^2,
\end{equation}
and that the time integrator is also second order, for instance Crank--Nicolson, so that the global
time discretization error is
\begin{equation}
\|u_h(t_n)-u_h^n\|\lesssim \Delta t^2.
\end{equation}
Balancing the spatial and temporal errors suggests the choice
\begin{equation}
\Delta t\sim h.
\end{equation}

Assume further that each time step requires $O(1)$ elliptic solves and that an optimal linear-cost
solver is available, so that the work per step is
\begin{equation}
O(N_h),\qquad N_h\sim h^{-d},
\end{equation}
up to logarithmic factors. Then the total classical cost over a time interval of length $t$ is
\begin{equation}
\mathrm{Cost}_{\mathrm{classical}}
\sim
\frac{t}{\Delta t}\,N_h
\sim
t\,h^{-d-1}.
\end{equation}

By Prop. \ref{prop:gaussian-zeno-stokes-new}, the Gaussian-Zeno quantum evolution has cost
\begin{equation}
    \widetilde O\!\left(\frac{\sqrt t}{h^2}\right)
\end{equation}
for fixed precision, when the block-encoding cost is dominated by the QSVT construction of the
divergence-free projector.  Thus, at the level of the evolution step alone, the comparison is
\begin{equation}
\frac{\sqrt t}{h^2}\quad \text{vs.}\quad \frac{t}{h^3}
\qquad\text{in two space dimensions,}
\end{equation}
and
\begin{equation}
\frac{\sqrt t}{h^2}\quad \text{vs.}\quad \frac{t}{h^4}
\qquad\text{in three space dimensions.}
\end{equation}
This suggests a possible asymptotic advantage in the $h$-dependence of the simulation stage, and
also a square-root dependence on the final time in the semigroup block-encoding.  A full end-to-end
comparison, however, would also need to account for state preparation, postselection or amplitude
amplification, measurement, and data-access costs, which are not analyzed here.

\subsection{Numerical tests}

We verify the unitary dilation for the 2D incompressible Stokes equations on $\Omega = [0,1]^2$ with Dirichlet boundary conditions by using the Marker-and-Cell (MAC) staggered grid.

\paragraph{MAC Staggered Grid.}
The marker-and-cell (MAC) staggered grid is a finite-difference discretization designed for incompressible flow problems such as the Stokes and Navier--Stokes equations. The related works for MAC staggered grid methods can be found in \cite{harlow1965numerical, mckee2008mac, rui2017stability, li2023error, chen2016finite, han1998new}. The pressure of the flow is typically stored at cell centers, while the velocity components are placed on cell faces. In two dimensions, for example, the horizontal velocity $u_1$ is stored at the midpoints of vertical cell edges, and the vertical velocity $u_2$ is stored at the midpoints of horizontal cell edges, whereas the pressure $p$ is stored at the center of each grid cell. This staggering is advantageous because the discrete divergence and gradient operators naturally couple neighboring pressure and velocity unknowns, and it helps avoid the spurious pressure oscillations that often occur on collocated grids. As a result, the MAC grid provides a stable and physically natural discretization of the incompressibility constraint $\nabla\cdot \bu=0$, with the discrete divergence mapping face-centered velocities to cell-centered scalar quantities and the discrete gradient mapping cell-centered pressures back to face-centered forces.

\begin{figure}[ht]
\centering
\includegraphics[width=0.5\linewidth]{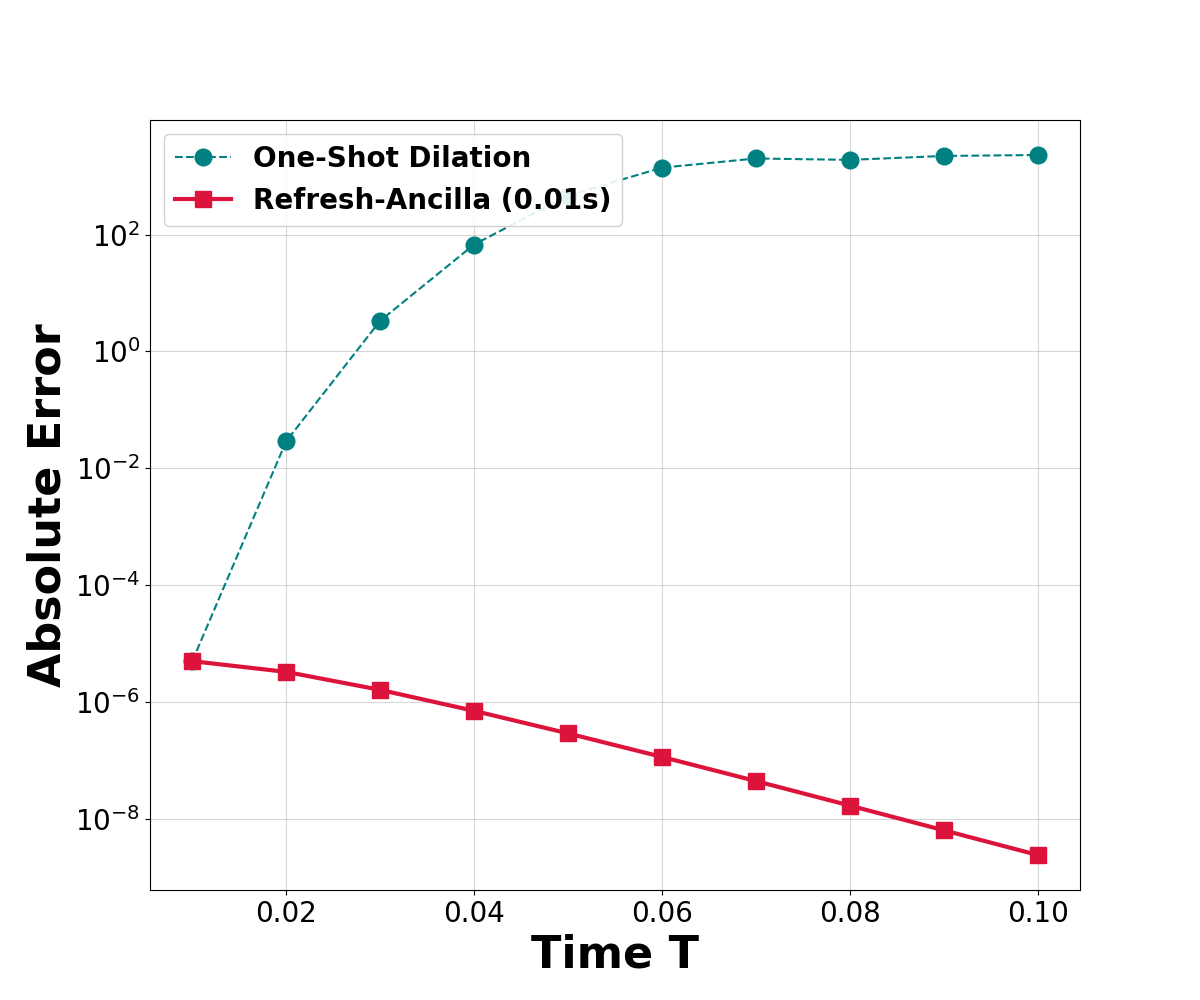}
\caption{Error from the dilation technique described in \cref{section: math_form}. Here we pick the ancilla Hilbert space to be $\C^{20}$ and the projection $p^* = 5\times 10^{-6}$.}
\label{fig: refresh_dil}
\end{figure}

\cref{fig: refresh_dil} plots  $\|\bm x(T) - \bm y(T) \|_{\ell_2}$, where $\bm x$ contains all the discretized values from \cref{eq:stokes_strong_application} and $\bm y$ is the vector containing the values from dilation method. One plots the error between the benchmark and the dilation, and another plots the dilation but perform ancilla refresh for every time step. \cref{fig: original_sys} plots the velocity and pressure original system \cref{eq:stokes_strong_application} after the MAC discretization. \cref{fig: dilated_sys} verifies \cref{thm:dilation} under short time period, meanwhile for long time we can do ancilla refreshment as in \cref{fig: refresh_dil}. 

\begin{figure}[ht]
    \centering
    \includegraphics[width=0.7\linewidth]{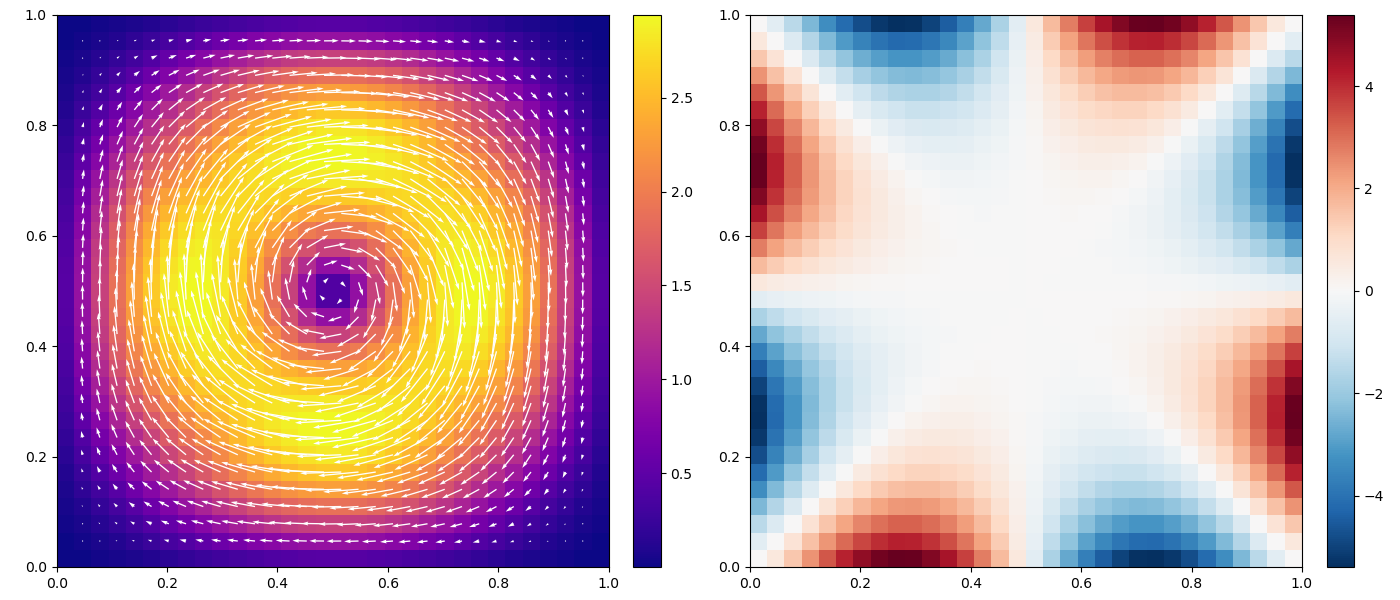}
    \caption{Velocity and pressure from the original discretized Stokes equation. Here we make the grid number be $32\times 32$, final time $t=0.001$ and the initial velocity is 
    $\bm u_0 = \mat{u_1(0) & u_2(0)}^T = \mat{-\pi\sin^2(\pi x)\sin(2\pi y) & \pi \sin^2(\pi y)\sin(2\pi x) }^T.$}
    \label{fig: original_sys}
\end{figure}

\begin{figure}[ht]
    \centering
    \includegraphics[width=0.7\linewidth]{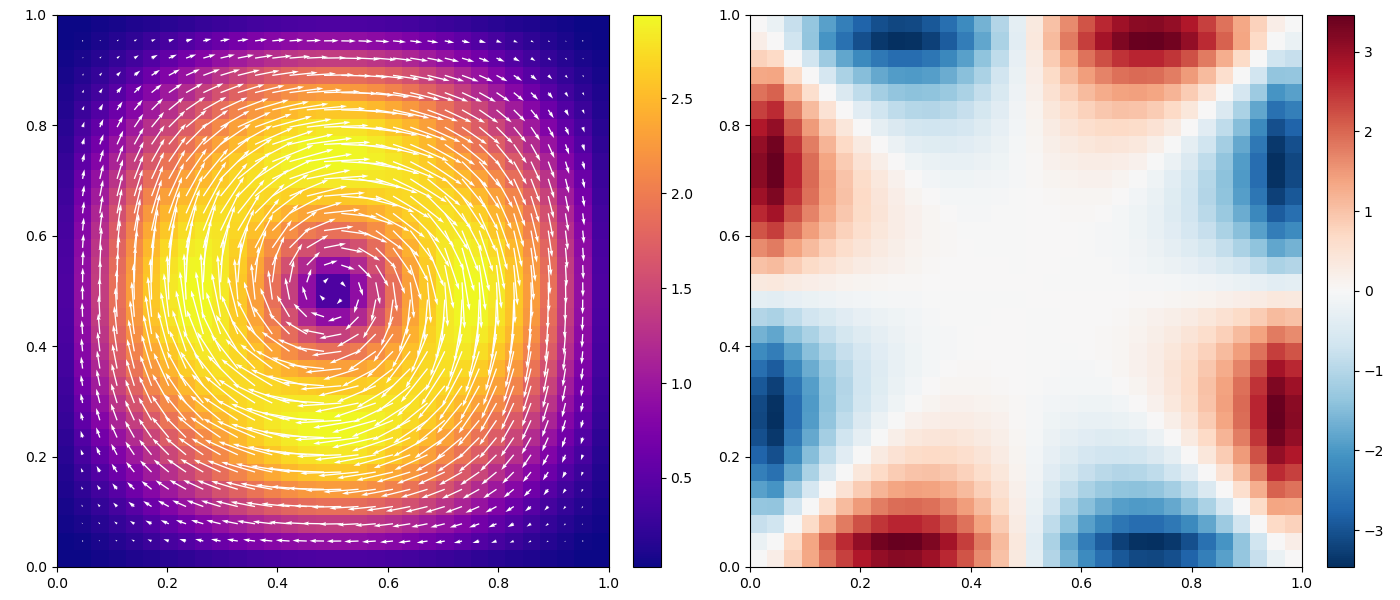}
    \caption{Recovered velocity and pressure from the dilated DAE. \cref{fig: dilated_sys} shows the recovery of the velocity and pressure from the original discretized DAE. We pick the ancilla space dimension $M=65$.}
    \label{fig: dilated_sys}
\end{figure}

\section{Summary and Discussions}

In this work, we have proposed a framework for quantum algorithms for linear differential-algebraic
equations (DAEs) by reformulating constrained dynamics as projected Schr\"odinger evolution on an
enlarged space. The main conceptual novelty is the identification of an intrinsic connection
between DAE constraints and quantum Zeno dynamics. In the dilated formulation, the physical
solution is recovered as the evolution inside an invariant constrained subspace selected by an
orthogonal projector. Thus, the constraint mechanism in the original DAE is represented, after
dilation, by a Zeno-type projected Hamiltonian dynamics. This reveals a natural bridge between
constrained numerical dynamics and Hamiltonian-based quantum computation.

Once the problem is written in this form, it becomes accessible to standard quantum algorithmic
tools. In particular, the projected dynamics can be expressed in terms of block-encodings of the
lifted Hamiltonian and the associated constraint projector, and can then be simulated using modern
Hamiltonian simulation techniques. From this perspective, the dilation is not merely a formal
embedding; it provides a mechanism for turning a constrained, generally nonunitary system into a
unitary projected dynamics compatible with quantum algorithms.

A particularly interesting application is to PDE discretizations. Among these, the unsteady
incompressible Stokes equations provide a natural test problem because they are fundamental in
fluid mechanics and central to numerical analysis. Their PDE origin is also important
algorithmically: as the mesh is refined, the associated discrete differential operators are not
uniformly bounded, with diffusive scales typically growing like $h^{-2}$. Thus, unlike
finite-dimensional DAE examples with uniformly bounded generators, the Stokes problem forces one
to confront the mesh dependence of Hamiltonian simulation, constraint projection, and
block-encoding construction.

For Stokes, the Zeno-reduced dynamics have an additional structure that goes beyond the generic
DAE formulation.  The reduced positive operator satisfies
\[
    S_h=-\Pi_h\Delta_h\Pi_h=(G_h\Pi_h)^\dagger(G_h\Pi_h),
\]
and therefore admits a first-order square-root Hamiltonian
\[
    B_h=
    \begin{pmatrix}
        0&\Pi_hG_h^\dagger\\
        G_h\Pi_h&0
    \end{pmatrix}.
\]
This Hamiltonian is itself a projected, or Zeno-compressed, Dirac-type operator.  We showed that
the Stokes semigroup can be represented by a Gaussian moment dilation,
\[
    e^{-tS_h}
    =
    (\langle 0|\otimes I)(\langle g|\otimes I)
    e^{-i\sqrt t\,F\otimes B_h}
    (|g\rangle\otimes I)(|0\rangle\otimes I),
\]
or equivalently by a Gaussian superposition of unitary evolutions generated by $B_h$.  This
replaces the second-order diffusive scale $\|S_h\|=\Theta(h^{-2})$ by the first-order scale
$\|B_h\|=\Theta(h^{-1})$ in the Hamiltonian simulation component.

Our current assessment suggests a modest but meaningful quantum advantage at the level of the
time-evolution step. In the generic sparse-access model, constructing the incompressibility
projector from the divergence operator by QSVT still contributes a factor $\widetilde O(h^{-1})$,
because $\|D_h\|=\Theta(h^{-1})$ while the discrete inf-sup gap is assumed to be
mesh-independent. Combining this projector cost with the Gaussian-LCHS simulation of the
square-root Hamiltonian gives the evolution-stage scaling
\[
    \widetilde O(h^{-2}\sqrt t)
\]
for fixed precision, up to the usual postselection factor needed to prepare the normalized
dissipative state.  

We view the complexity discussion in this paper as a first step rather than a definitive
separation from classical algorithms. The results indicate that constrained PDE dynamics fit
naturally into quantum Zeno and Hamiltonian-simulation frameworks, and that additional PDE
structure may be used to improve the mesh and time dependence of the quantum simulation stage.
Further progress may come from sharper block-encoding constructions, more efficient use of
filtered or spectral projectors, improved state preparation and measurement strategies, and more
detailed comparisons with high-performance classical solvers.

We hope that the framework developed here provides a useful starting point for broader
investigations. It would be interesting to test these ideas on other classes of DAEs, other
structure-preserving discretizations, and other constrained systems from fluid mechanics,
mechanics, and scientific computing. More broadly, the comparison between the universal
moment-matching dilation for general nonunitary constrained dynamics and the Gaussian
moment dilation for factorizable dissipative generators suggests a hierarchy of quantum
simulation strategies: general but potentially diffusive-scale methods on one hand, and more
specialized square-root methods on the other. Understanding when such structures arise, and how
they interact with constraints, remains an important direction for future work.
\section{Acknowledgments}
This work was supported by NSF Grant DMS-2411120. The authors used ChatGPT and Gemini for language polishing and programming assistance.

\bibliographystyle{unsrt}
\bibliography{ham, zeno, dae}

\end{document}